\documentclass[a4paper,twocolumn,11pt,accepted=2025-01-23]{quantumarticle}
\pdfoutput=1
\usepackage[utf8]{inputenc}
\usepackage[english]{babel}
\usepackage[T1]{fontenc}
\usepackage{amsmath,amssymb}
\usepackage{amsfonts}
\usepackage{amsthm}
\usepackage{hyperref}
\usepackage[numbers]{natbib}

\usepackage{tikz}
\usepackage{lipsum}
\usepackage{algorithm2e}
\usepackage[capitalise]{cleveref}
\usepackage[noend]{algpseudocode}
\usepackage{braket}

\newtheorem{mytheorem}{Theorem}
\newtheorem{mylemma}{Lemma}
\newtheorem{mycorollary}{Corollary}
\newtheorem{myproposition}{Proposition}
\newtheorem{mydefinition}{Definition}

\def\Tr{\mathrm{Tr}}

\newcommand{\mcN}{\mathcal{N}}
\newcommand{\mcU}{\mathcal{U}}
\newcommand{\mcV}{\mathcal{V}}
\newcommand{\mcH}{\mathcal{H}}
\newcommand{\mcB}{\mathcal{B}}
\newcommand{\mc}{\mathcal}
\newcommand{\vect}[1]{\boldsymbol{#1}}
\newcommand{\eps}{\ensuremath{\varepsilon}}

\newcommand{\bes} {\begin{subequations}}
\newcommand{\ees} {\end{subequations}}
\newcommand{\beq}{\begin{equation}}
\newcommand{\eeq}{\end{equation}}

\def\a{\alpha}

\def\k{\kappa}

\def\r{\rho}     
\def\s{\sigma}

\def\Ph{\mcN}

\def\O{\Omega}

\def\>{\rangle}
\def\<{\langle}
\def\Tr{\mathrm{Tr}}
\def\Pr{\mathrm{Pr}}

\newcommand{\ketb}[2]{|{#1}\>\!\<#2|}

\begin{document}

\title{Beyond unital noise in variational quantum algorithms: noise-induced barren plateaus and limit sets}

\author{Phattharaporn Singkanipa}
\affiliation{Department of Physics and Center for Quantum Information Science \& Technology,
University of Southern California, Los Angeles, CA 90089, USA}

\author{Daniel A. Lidar}
\affiliation{Departments of Electrical Engineering, Chemistry, Physics \& Astronomy, and Center for Quantum Information Science \& Technology,
University of Southern California, Los Angeles, CA 90089, USA}

\maketitle

\begin{abstract}
  Variational quantum algorithms (VQAs) hold much promise but face the challenge of exponentially small gradients. Unmitigated, this barren plateau (BP) phenomenon leads to an exponential training overhead for VQAs. Perhaps the most pernicious are noise-induced barren plateaus (NIBPs), a type of unavoidable BP arising from open system effects, which have so far been shown to exist for unital noise maps. Here, we generalize the study of NIBPs to more general completely positive, trace-preserving maps, investigating the existence of NIBPs in the unital case and a class of non-unital maps we call Hilbert-Schmidt (HS)-contractive. The latter includes amplitude damping. We identify the associated phenomenon of noise-induced limit sets (NILS) of the VQA cost function and prove its existence for both unital and HS-contractive non-unital noise maps. Along the way, we extend the parameter shift rule of VQAs to the noisy setting. We provide rigorous bounds in terms of the relevant variables that give rise to NIBPs and NILSs, along with numerical simulations of the depolarizing and amplitude-damping maps that illustrate our analytical results.
\end{abstract}

\section{Introduction}

Variational quantum algorithms (VQAs) are promising applications of quantum computing in the NISQ era~\cite{Preskill2018,Cerezo2021,Endo2021,McClean:2016aa}. These algorithms leverage a customizable quantum circuit design, integrating both quantum and classical computation capabilities. Using parameterized quantum circuits, they compute problem-specific cost functions, followed by classical optimization to iteratively update the parameters. This hybrid quantum-classical optimization process continues until predefined termination criteria are met.

Previous studies have demonstrated that VQA circuits operable within the existing noise levels and hardware connectivity limitations of the NISQ era, already find applications across diverse domains, such as quantum optimization~\cite{farhi2014quantum,Moll2018,Wang2018}, quantum optimal control~\cite{Li:2017aa}, linear systems~\cite{Bravo-Prieto2019,Huang2021,Xu2021-lin}, quantum metrology~\cite{Koczor2020,Meyer2021}, quantum compiling~\cite{Khatri2019,Sharma2020}, quantum error correction~\cite{qvector,Xu2021}, quantum machine learning~\cite{Mitarai2018,Farhi2018} and quantum simulation~\cite{Peruzzo:2014aa,Kandala:2017aa,Cerezo2022}. Moreover, VQA has been established as a universal model of quantum computation~\cite{Biamonte2021}. 

Despite their comparable computational power to other quantum models and demonstrated advantages, VQAs exhibit inherent constraints that present scalability challenges for problems of arbitrary scale. Specifically, VQAs for random circuits suffer from exponentially vanishing gradients, commonly referred to as the Barren Plateau (BP) phenomenon~\cite{McClean2018,Holmes2022,Cerezo2021-bp,Ortiz2021,Cerezo2021-highbp,Arrasmith2021,CerveroMartin2023barrenplateausin,Wang2021,Arrasmith:2022aa,fefferman2023effect}. This phenomenon renders the parameter training step asymptotically impossible for circuits with a sufficiently large number of qubits $n$, even at shallow circuit depth.

Here, we study noise-induced barren plateaus (NIBPs), which emerge under decoherence-induced noise~\cite{Breuer:book}. NIBPs were previously shown to be present in sufficiently deep circuits subjected to unital maps~\cite{Wang2021}. This holds true even in constant-width or non-random circuits. Alternatively, NIBPs exist under strictly contracting noise maps when the parameter shift rule (PSR)~\cite{Mitarai2018,Schuld2019} is applicable~\cite{schumann2023emergence}. 

Unlike other BP types, for which mitigation strategies have been proposed~\cite{Volkoff2021,Holmes2022,Grant2019,Zhang2020,Pesah2021,Patti2021,Bharti2021,cichy2023nonrecursive,Wiersema2023,Mele2022,LIU2022128169}, it remains unclear whether NIBPs can be similarly mitigated. 
Experimental investigations on small systems have suggested that error mitigation (EM) techniques enable VQAs to more closely approach the true ground-state energy~\cite{Rosenberg2022}. Clifford Data Regression has proven effective in mitigating errors and reversing the concentration of cost function values~\cite{Czarnik2021errormitigation}. Nevertheless, it is noteworthy that the majority of EM protocols do not enhance trainability or even exacerbate the lack of trainability. Additionally, post-processing expectation values of noisy circuits is not advantageous in the context of NIBP~\cite{Wang2021-EM}. Previous work suggests that stochastic noise could be helpful for training VQAs \cite{liu2023stochastic}; an interesting open question that remains is whether there is an intermediate noise regime where we can train VQAs by exploiting noise. 

In this work, we extend the study of NIBPs to completely positive trace-preserving (CPTP) maps, including both unital maps and a class of non-unital maps we call HS-contractive. A rigorous definition of this class is given in \cref{def:HS-contractive} below, but intuitively, this is the class of maps under which the Bloch vector (or more generally, the coherence vector) is shrunk before it is shifted, just as in the case of the amplitude damping map. We analytically derive the scaling of the cost function gradient as a function of circuit width $n$, circuit depth $L$, and noise strength. We find that HS-contractive non-unital noise need not necessarily give rise to NIBPs, but instead exhibits a different phenomenon, which we refer to as a noise-induced limit set (NILS). Moreover, we simplify the NIBP derivation compared to Ref.~\cite{Wang2021}, guided by the intuition gained by considering the effect of noise on the single-qubit Bloch sphere. We generalize this to $n$-qubit systems via the coherence vector and compute derivatives of the cost function via the PSR. In addition, we investigate the applicability of the PSR under control noise and random unitary noise and assess the impact of these noise types on the bounds we derive. We find analytical expressions for the dependence of the circuit depth $L$ on relevant noise and circuit parameters that give rise to NIBP and NILS. Our analytical results are supported by numerical simulations.

This paper is organized as follows. Background results regarding VQA, PSR, the coherence vector, and characteristics of CPTP noise maps are presented in \cref{sec:pre}. \cref{sec:psr-w-noise} extends the PSR analysis to scenarios involving noise. We study the effects of HS-contractive non-unital and general unital noise in \cref{sec:NILS} and \cref{sec:NIBP}, respectively. In the unital case, we reprove that NIBP is always present. In the HS-contractive non-unital case, we find new results, in particular the phenomenon of NILS. Our theoretical findings are supported by numerical simulations in \cref{sec:sim}. We summarize our findings in \cref{sec:conclusions}.

\section{Preliminaries}
\label{sec:pre}

In this section, we review pertinent technical details and establish the notation we use to derive our results.

\subsection{VQA and PSR}

We adopt the variational quantum algorithm (VQA) framework of Ref.~\cite{Wang2021} and consider a general class of parameterized unitary ansatzes:
\bes
\label{eq:1}
\begin{align}
\label{eq:1a}
U(\vect{\theta}) &= \Pi_{l=L}^1 U_l(\vect{\theta}_l),\\
\label{eq:1b}
U_l(\vect{\theta}_l) &= \Pi_{m=G_l}^1e^{-\frac{i}{2}\theta_{lm}H_{lm}} W_{lm},
\end{align}
\ees
where $L$ is the circuit depth and the $U_l(\vect{\theta}_l)$ are unitaries sequentially applied by layers. The $l$'th layer consists of $G_l$ gates: the unparametrized gates denoted by $W_{lm}$ (such as CNOT) and the gates generated by dimensionless Hamiltonians denoted by $H_{lm}$ (the $m$'th gate in the $l$'th layer). In writing the various gates in \cref{eq:1} we implicitly assume that they are in a tensor product with the identity operator acting on all the qubits that are not explicitly labeled. The set $\vect{\theta} = \{\vect{\theta}_l\}_{l=1}^L$ consists of vectors of dimensionless continuous parameters $\vect{\theta}_l = \{\theta_{lm}\}_{m=1}^{G_l}$ that are optimized to minimize a cost function $C_\Omega$ expressed as the expectation value of an operator $\O$:
\beq
C_\Omega(\vect{\theta}) =\Tr[\Omega\ \mathcal{U}(\vect{\theta})(\rho_0)] .
\label{eq:C}
\eeq
Here,
\beq
\mathcal{U}(\vect{\theta})(\rho_0) \equiv U(\vect{\theta})\rho_0 U^\dag(\vect{\theta}) = \rho(\vect{\theta})
\label{eq:mcU}
\eeq 
is the unitary superoperator acting on the initial state $\rho_0$.
An important special case, which we focus on, is when $\Omega=H$, the ``problem Hamiltonian'' whose energy one is trying to minimize, and in this case we simply write $C$ for the cost function.

For $n$ qubits, we can always parametrize the traceless gate-generating Hamiltonians as
\beq
H_{lm} = \sum_{j=1}^{d^2-1} h_{lmj} P_j = \vect{h}_{lm}\cdot \vect{P},
\label{eq:Hlm}
\eeq
where $P_j\in\{I,\sigma^x,\sigma^y,\sigma^z\}^{\otimes n}$ is a Pauli string, i.e., a tensor product of up to $n$ Pauli matrices, $P_0=I^{\otimes n}$, $I$ is the identity operator, and $d=2^n$. We assume that the $P_j$'s are ordered such that $j$ increases with the Hamming weight of the Pauli string, i.e., the number of non-identity terms in $P_j$ (the manner in which $j$ increases at fixed Hamming weight does not matter for our purposes), and $\vect{P} = (P_1,\dots, P_{d^2-1})$.

In most cases of interest, the $h_{lmj}\in\mathbb{R}$ vanish for strings involving more than two Pauli matrices, i.e., the Hamiltonians are two-local.
This framework includes the Quantum Approximate Optimization Algorithm or Quantum Alternating Operator Ansatz (QAOA)~\cite{farhi2014quantum,Wang2018}, where $H_{l}=H_1\delta_{l,\text{odd}}+H_2\delta_{l,\text{even}}$ $\forall l,m$ with $[H_1,H_2]\ne 0$,
the Unitary Coupled Cluster (UCC) ansatz~\cite{Lee:2019aa}, where the $h_{lmj}$ are coefficients derived from one- and two-electron integrals, which is used in the Variational Quantum Eigensolver (VQE) algorithm~\cite{Peruzzo:2014aa} with applications in quantum chemistry~\cite{Cao:2019aa}, and the Hardware Efficient VQE Ansatz, which tries to minimize the circuit depth (i.e., the set of non-zero $\theta_{lm}$) given a predefined gate-set tailored to particular hardware~\cite{Kandala:2017aa}.

The parameter shift rule (PSR) is frequently used in evaluating derivatives of cost functions in VQAs~\cite{Li:2017aa,Mitarai2018,Schuld2019}. For a cost function $C(\vect{\theta})$ as in \cref{eq:C}, the PSR states that (see, e.g.,~\cite[Table~2]{Schuld2019}):
\begin{equation}
    \begin{aligned}
        \frac{\partial C(\vect{\theta})}{\partial \theta_{lm}}
        &=\frac{1}{2}[C(\vect{\theta}+\vect{\theta}_{lm}^{\pi/2})-C(\vect{\theta}-\vect{\theta}_{lm}^{\pi/2})],
    \end{aligned}
    \label{eq:psr}
\end{equation}
where 
\beq
\vect{\theta}_{lm}^{\pi/2}= \frac{\pi}{2} \hat{e}_{lm} 
\eeq 
and $\{\hat{e}_{lm}\}$ are standard unit vectors [i.e., the $(l,m)$th component of $\vect{\theta}_{lm}^{\pi/2}$ is $\pi/2$ and the rest are zero]. We reprove this result in \cref{sec:noiseless-PSR}. The essential point is that we can compute the derivative by means of a finite difference.

\subsection{Nice operator basis and Schatten $p$-norms}

Consider a Hilbert space $\mc{H}$ of dimension $d<\infty$. The space of bounded linear operators acting on $\mc{H}$ is denoted $\mcB(\mc{H})$. Let $\mc{M}(d, \mathbb{F})$ denote the vector space of $d\times d$ matrices with coefficients
in $\mathbb{F}$, where $\mathbb{F} \in \{\mathbb{R}, \mathbb{C}\}$. For our purposes it suffices to identify $\mcB(\mc{H})$ with $\mc{M}(d,\mathbb{C})$. The Hilbert-Schmidt inner product is $\<A,B\> \equiv \Tr(A^\dag B)$ for any two operators $A,B\in \mcB(\mc{H})$.

We define a ``nice operator basis'' as a set $\{F_j\}_{j=0}^{d^2-1}\in\mathcal{B}(\mc{H})$, where $F_0=\frac{1}{\sqrt{d}}I$, $\Tr(F_{j})=0$ $\forall j\ge 1$, that in addition satisfies the following properties:
\begin{equation}
    F_j=F_j^\dag, \ \ \<F_j, F_k\>=\Tr(F_j F_k)=\delta_{jk} \quad \forall j,k . 
    \label{eq:nice-op}
\end{equation}
The normalized Pauli strings $\{\frac{1}{\sqrt{d}}P_j\}_{j=0}^{d^2-1}$ (where $d=2^n$) are a convenient explicit choice for the nice operator basis. 
Another convenient choice is the set of generalized $d\times d$ Gell-Mann matrices~\cite{Gell-Mann1962,generalized-Gell-Mann}, normalized such that $\Tr(F_j F_k)=\delta_{jk}$ is satisfied. 

Let $|A| \equiv \sqrt{A^\dag A}$. Recall that the Schatten $p$-norm $\|A\|_p$ is given, for $1 \leq p < \infty$, by the $p$-norm of the singular values $\sigma_i$ of the operator $A\in\mcB(\mc{H})$:
\beq
\|A\|_p = \Tr(|A|^p)^{1/p} = ( \sum_i \sigma_i^p )^{1/p}  .
\label{eq:Schatten-p}
\eeq
$\|A\|_1 =\Tr(|A|)$ is the trace norm (sum of the singular values), $\|A\|_2 = \sqrt{\<A,A\>} = \sqrt{\Tr(|A|^2)}$ is the Hilbert-Schmidt or Frobenius norm, and $\|A\|_\infty$ is the operator norm (largest singular value). Without risk of confusion, we use $\|A\|$ to denote $\|A\|_\infty$ and also use $\|\vect{v}\|$ to denote the Euclidian norm (i.e., $2$-norm) of any \emph{vector} $\vect{v}\in\mc{H}$ from hereon.

\subsection{The coherence vector}

Quantum states are represented by density operators $\rho\in \mcB_+(\mcH)$ (positive trace-class operators acting on $\mc{H}$) with unit trace: $\Tr\r =1$.
Elements of $\mathcal{B}[\mathcal{B}(\mc{H})]$, i.e., linear transformations $\mcN:\mathcal{B}(\mc{H})\mapsto \mathcal{B}(\mc{H})$, are called superoperators, or maps. 

Complete positivity of a superoperator $\mcN$ is equivalent to the statement that  $\mcN$ has a Kraus representation~\cite{Kraus:83}: $\forall X\in \mcB(\mc{H})$, 
\beq
\mcN(X) = \sum_\alpha K_\alpha X K_\alpha^\dag ,
\label{eq:KOSR}
\eeq
where the $\{K_\a\}$ are called Kraus operators. When they satisfy $\sum_\a K_\a^\dag K_\a = I$, the map $\mcN$ is trace-preserving. 

The density operator can be expanded in an arbitrary nice operator basis as
\begin{equation}
    \rho=\frac{1}{d}I + \sum_{j=1}^{d^2-1}v_j F_j = \frac{1}{\sqrt{d}}F_0+\vect{v}\cdot\vect{F},
    \label{eq:rho-nice}
\end{equation}
where $\vect{F}=\{F_1,\dots,F_{d^2-1}\}^T$, and $\vect{v}=\{v_1,\dots,v_{d^2-1}\}$ is called the \emph{coherence vector}. 

We summarize two well-known facts about the coherence vector. First, the purity 
\beq
P \equiv \Tr\r^2 = \<\r,\r\> = \|\r\|_2^2
\eeq
is bounded by
\begin{align}
0\leq \|\vect{v}\|= \left(P-\frac{1}{d}\right)^{1/2} &\leq \left(1-\frac{1}{d}\right)^{1/2} < 1 \ .
\label{eq:v-bounded}
\end{align}
See \cref{app:purity-proof} for a proof.

Second, 
let $\mc{M}(d, F)$ denote the vector space of $d\times d$ matrices with coefficients in the field $F$.

\begin{myproposition}
\label{prop:cohvec-affine}
The CPTP map $\r'=\mcN(\rho)$ is equivalent to the affine coherence vector transformation
    \beq
    \label{eq:v'}
    \vect{v}' = M\vect{v}+\vect{c} ,
    \eeq
where $M\in \mc{M}(d^2-1, \mathbb{R})$ and $\vect{c}\in \mathbb{R}^{d^2-1}$ have elements given by
\bes
\label{eq:M-c}
    \begin{align}
      \label{eq:M-c-1} 
              M_{ij} &=\< F_i,\mcN(F_j)\> = \sum_\alpha
        \Tr(F_i K_\alpha F_j K_\alpha^\dag)
        \\
      \label{eq:M-c-2} 
        c_i &=\frac{1}{d} \< F_i,\mcN(I)\> = \frac{1}{{d}}\sum_\alpha
        \Tr(F_i K_\alpha K_\alpha^\dag )\ .
            \end{align}
\ees
\end{myproposition}
See \cref{app:cM-proof} for a proof.

The Gell-Mann matrices reduce to the standard Pauli matrices for $d=2$, normalized such that $\vect{F}=\vect{\sigma}/\sqrt{2}=(\sigma^x,\sigma^y,\sigma^z)/\sqrt{2}$.
Therefore, in the case of single qubit, we can write $\rho$ in the well-known form
\begin{equation}
        \rho =\frac{1}{2}(I+\Bar{\vect{v}}\cdot\vect{\sigma})=\frac{1}{\sqrt{2}}F_0+{\vect{v}}\cdot\vect{F},
\end{equation}
where ${\vect{v}}=\Bar{\vect{v}}/\sqrt{2}$. 
Note that $\|\Bar{\vect{v}}\| \leq 1$, which is the convention for the Bloch sphere representation. We avoid this normalization and instead use the nice operator basis convention from here on, even for $d=2$, so that $\|\vect{v}\| \leq 1/\sqrt{2}$ [\cref{eq:v-bounded}].

\subsection{Unital maps}
A unital map $\mcN$ is defined as satisfying $\mcN(I)=I$, hence $\sum_\alpha K_\alpha K_\alpha^\dag =I$. 

\begin{mylemma}
\label{lemma1}
Unital CPTP maps are purity non-increasing: $P'\leq P$, where $P$ and $P'$ are, respectively, the purity of $\r$ and $\r' = \mcN(\r)$. Equality holds for all $\rho$ iff the map is unitary.
\end{mylemma}
See \cref{app:lemma1-proof} for a proof.

\begin{mylemma}
\label{lemma2}
For unital CPTP maps $\mcN$ we have:
\bes
\begin{align}
    \label{eq:lemma2-1}
\vect{c} &= \vect{0} ,\\
    \label{eq:lemma2-2}
\|\vect{v}'\| &= \|M\vect{v} \| \le \|\vect{v} \|  .
\end{align}
\ees
Equality in \cref{eq:lemma2-2} holds for all $\vect{v}$ iff $\mcN$ is unitary, in which case $M$ is norm-preserving (hence orthogonal).
\end{mylemma}

\begin{proof}
To prove \cref{eq:lemma2-1}, note that it follows from \cref{eq:M-c} that $c_i = \frac{1}{{d}}
        \Tr[F_i \sum_\alpha K_\alpha K_\alpha^\dag ] = \frac{1}{{d}} \Tr[F_i]=0$.
        
To prove \cref{eq:lemma2-2}, note that from \cref{eq:v-bounded,eq:v'} we have $ \|M\vect{v} \| =  \|\vect{v}'\| = \sqrt{P'-1/d}$, where $P'=\Tr[(\r')^2]$, $\r'=\mcN(\r)$. If $\|M\vect{v} \|> \|\vect{v} \| = \sqrt{P-1/d}$ then $P'> P$, which contradicts \cref{lemma1}. Moreover, from \cref{lemma1}, $P'=P$ iff $\mcN$ is unitary. Since $P'=P$ is equivalent to $ \|M\vect{v} \| =  \|\vect{v}'\| = \|\vect{v}\|$, we have equality for all $\vect{v}$ iff $\mcN$ is unitary.
    \end{proof}

\subsection{Non-unital maps}
\label{sec:non-unital}

From here on, when we consider non-unital noise maps we restrict our analysis to the following class, which are contractive under the Hilbert-Schmidt (HS) norm:

\begin{mydefinition}
\label{def:HS-contractive}
A (finite-dimensional) map $\mathcal{N}$ is called \emph{HS-contractive} if $\exists r < 1$ s.t. for all states $\rho_1\ne \rho_2$ we have $\|\mathcal{N}(\rho_1) - \mathcal{N}(\rho_2)\|_2 \leq r \|\rho_1 - \rho_2\|_2$.
\end{mydefinition}

This definition of contractivity is different from the standard one for CPTP maps, which are well known to be contractive under the trace norm, i.e., $\|\mathcal{N}(\rho_1) - \mathcal{N}(\rho_2)\|_1 \leq \|\rho_1 - \rho_2\|_1$ for any pair of states $\rho_1,\rho_2$ and any CPTP map $\mc{N}$, including all non-unital maps~\cite{nielsen2010quantum}. Other notions of contractivity also exist, such as for general positive maps between matrix spaces that include the zero element, in which case a non-unital map is always non-contractive~\cite{Perez-Garcia:2006aa}; see \cref{app:norm} for details.

\begin{mylemma}
\label{lem:M<1}
A map is HS-contractive if and only if its matrix $M$ satisfies $\|M\|<1$.
\end{mylemma}

\begin{proof}
In terms of a nice operator basis and the corresponding coherence vector, we can write $\|\rho_1 - \rho_2\|_2 = \|(\vect{v}_1-\vect{v}_2)\cdot\vect{F}\|_2$. Let $\vect{v}=\vect{v}_1-\vect{v}_2$ and using the properties of $\{F_i\}$ in \cref{eq:nice-op}, we have $\|\vect{v}\cdot\vect{F}\|_2=\|\vect{v}\|$.
Let $\rho'=\mathcal{N}(\rho)$ and $\vect{v}'=M\vect{v}+\vect{c}$. We can write
$\|\rho'_1 - \rho'_2\|_2 = \|(\vect{v}'_1-\vect{v}'_2)\cdot\vect{F}\|_2= \|(M\vect{v}_1-M\vect{v}_2)\cdot\vect{F}\|_2=\|M\vect{v}\|$.

($\Rightarrow$) By definition of $\mathcal{N}$ being HS-contractive we have $\|M\vect{v}\| = \|\mathcal{N}(\rho_1) - \mathcal{N}(\rho_2)\|_2 \le r \|\rho_1 - \rho_2\|_2  = r \|\vect{v}\| < \|\vect{v}\|$, $\forall \vect{v}\neq\vect{0}$. This is true in particular for the vector $\vect{v}$ that achieves the supremum in $\sup_{\vect{v}\neq\vect{0}}\|M\vect{v}\|/\|\vect{v}\| = \|M\|$, which implies that $\|M\|<1$.
    
($\Leftarrow$) Let $r = \|M\|$. By assumption, $r < 1$. We have $\|M\vect{v}\| \leq \|M\|\|\vect{v}\|= r \|\vect{v}\|$ for any $\vect{v}$. Taking two arbitrary density matrices $\rho_1, \rho_2$ with respective coherence vectors $\vect{v}_1,\vect{v}_2$ and applying this to $\vect{v}=\vect{v}_1-\vect{v}_2$, we obtain $\|\mcN(\rho_1) - \mcN(\rho_2)\|_2 \le r\|\rho_1 - \rho_2\|_2$.
\end{proof}

Using the polar decomposition, let us decompose $M$ in \cref{eq:M-c} as $M=OS$, where $O$ is orthogonal and $S=|M|$ is positive semidefinite. Let $\sigma_{\max/\min} (M)$ denote the largest/smallest singular value of $M$ (the largest/smallest eigenvalue of $S$). We interpret $O$ as a rotation, $S$ as a dilation, and $\vect{c}$ [\cref{eq:v'}] as an affine shift.

\begin{mylemma}
\label{lemma3}
For HS-contractive non-unital CPTP maps $\mcN$, we have:
\bes
\begin{align}
& \vect{c}\ne \vect{0} \label{eq:lemma3-a}\\
& \|\vect{c}\| \leq 1/\sqrt{1-1/d} \label{eq:lemma3-b}\\
& \|M\|< 1 \label{eq:lemma3-c}\\
&\|M\vect{v} \| < \|\vect{v} \| , \forall \vect{v} \neq \vect{0} .
\label{eq:lemma3-d}
\end{align}
\ees
\end{mylemma}

\begin{mylemma}
\label{lemma3-qubit}
Any single-qubit ($d=2$) non-unital map is always HS-contractive.
\end{mylemma}

See \cref{app:lemma3-proof} for proofs of these two Lemmas.

\begin{mycorollary}
\label{cor1}
If $\mcN$ is an HS-contractive non-unitary CPTP map and $\mcU$ is unitary, with corresponding coherence vector transformations $\vect{v}' = M\vect{v}+\vect{c}$ and $\vect{v}' = O\vect{v}$, where $O$ is orthogonal, then for the maps $\mcN\circ\mcU$ and $\mcU\circ\mcN$ we have:
\beq
\|MO\vect{v}\| , \|OM\vect{v} \| < \|\vect{v} \| .
\eeq
\end{mycorollary}

\begin{proof}
This is an immediate consequence of \cref{lemma2} and \cref{lemma3}, since whether $\mcN$ is unital or HS-contractive non-unital we have $\|M\vect{v} \| < \|\vect{v} \|$, so that: $\|M(O\vect{v})\| < \|O\vect{v}\| = \|\vect{v}\|$, and $\|O(M\vect{v})\| = \|M\vect{v} \| < \|\vect{v}\|$.
\end{proof}

\section{Parameter shift rule in the presence of noise}
\label{sec:psr-w-noise}

The PSR given in \cref{eq:psr} is valid for closed systems undergoing unitary evolution without noise. This section presents a brief rederivation for the noiseless setting, which we then adapt to accommodate scenarios involving control noise and random unitary noise. We also bound the gradient of the cost function in both of these cases.

\subsection{The noiseless PSR case}
\label{sec:noiseless-PSR}

For simplicity (and w.l.o.g., but at the expense of increasing the circuit depth), let us assume that each of the gate Hamiltonians $H_{lm}$ [\cref{eq:Hlm}] is a single Pauli string, i.e., we can write the terms in \cref{eq:1b} as 
\beq
\exp(-\frac{i}{2}\theta_{lm} H_{lm}) = \exp(-\frac{i}{2}\theta_\mu P_{j(\mu)}) \equiv U(\theta_\mu) ,
\label{eq:idealU}
\eeq 
where $\mu = (l,m)$ is the location of the gate in the circuit (the $m$'th gate in the $l$'th layer), we wrote $j(\mu)$ since the Pauli string type depends on the location $\mu$, and we dropped the subscript $j$ on $U$ since, in the calculation below, the type of Pauli string will not matter. From now on we sometimes also write $j$ instead of $j(\mu)$ for notational simplicity. 

Recall \cref{eq:1} and let us write $\vect{\theta} = \{\vect{\theta}_a,\theta_\mu,\vect{\theta}_b\}$, where $\vect{\theta}_b$ and $\vect{\theta}_a$ collect the rotation angles before and after $\theta_\mu$, respectively. Thus, $\mcU(\vect{\theta}) = \mcU(\vect{\theta}_a)\circ \mcU(\theta_\mu)\circ\mcU(\vect{\theta}_b)$, where we used the unitary superoperator notation of \cref{eq:mcU}.

Anticipating the derivative with respect to $\theta_\mu$, we rewrite the cost function as follows:
\bes
\begin{align}
C(\vect{\theta}) &=\Tr[H \rho(\vect{\theta}) ]  = \Tr[H \mcU(\vect{\theta})(\rho_0) ] \\
&= \Tr[H \mcU(\vect{\theta}_a)\circ \mcU(\theta_\mu)\circ\mcU(\vect{\theta}_b)(\rho_0) ]\\
& = \Tr[\tilde{H}  \mcU(\theta_\mu)(\tilde{\rho})] ,
\label{eq:21c}
\end{align}
\ees
where 
\bes
    \begin{align}
\label{eq:tildeH}
        \tilde{H} &=\mcU^\dag(\vect{\theta}_a)(H) = U^\dag(\vect{\theta}_a) HU(\vect{\theta}_a)\\
 	\tilde{\rho} &= \mcU(\vect{\theta}_b)(\rho_0) .
    \end{align}
\ees

Thus, since $\partial_\theta \mcU(\theta)(\cdot) = -\frac{i}{2}\mcU(\theta)([P_j,\cdot])$ for $\mcU(\theta)(\cdot) = e^{-i\theta P_j/2} \cdot e^{i\theta P_j/2}$:
\begin{equation}
    \begin{aligned}
         \frac{\partial C(\vect{\theta})}{\partial \theta_\mu}
        &=-\frac{i}{2}\Tr[\tilde{H} \mcU(\theta_\mu)([P_{j(\mu)},\tilde{\rho}])] .
    \end{aligned}
    \label{eq:cost-diff}
\end{equation}

The following identity holds for $U(\theta) = \exp(-i\theta P_j/2)$, where $P_j$ is an arbitrary Pauli string~\cite{Mitarai2018}:
\begin{equation}
    [P_j,\rho]=i\mcU\left(\frac{\pi}{2}\right)(\rho) -i \mcU\left(-\frac{\pi}{2}\right)(\rho).
    \label{eq:pauli-identity}
\end{equation}
Using this identity in \cref{eq:cost-diff} we have:
\bes
    \label{eq:C-terms}
        \begin{align}
        \notag
&        \frac{\partial C(\vect{\theta})}{\partial \theta_\mu}\\ 
        &=\frac{1}{2}\Tr\left[ \tilde{H} \mcU(\theta_\mu)\left(\mcU\left(\frac{\pi}{2}\right)(\tilde{\rho}) -\mcU\left(-\frac{\pi}{2}\right)(\tilde{\rho}) \right) \right]\\ 
        &=\frac{1}{2}\Tr[ \tilde{H} \mcU\left(\theta_\mu+\frac{\pi}{2}\right)(\tilde{\rho}) ] \notag\\
        &\qquad\qquad -\frac{1}{2}\Tr[\tilde{H} \mcU\left(\theta_\mu-\frac{\pi}{2}\right)(\tilde{\rho}) ]\\
        &=\frac{1}{2}[C(\{\vect{\theta}_a,\theta_\mu+\frac{\pi}{2},\vect{\theta}_b\})-C(\{\vect{\theta}_a,\theta_\mu-\frac{\pi}{2},\vect{\theta}_b\})],
    \label{eq:C-terms-c}
    \end{align}
\ees
where in the last line we used \cref{eq:21c}. This is the closed system PSR, \cref{eq:psr}. 

\subsection{Control noise}
\label{sec:control-noise}

Now consider adding a small perturbation to the ideal Pauli generator of a gate, i.e., $P_j\mapsto P_j+ A_j$, such that $\|A_j\| \ll 1$. 
In analogy to \cref{eq:Hlm}, we can decompose $A_j$ in the Pauli basis such that 
\beq
A_j=\sum_{k=0}^{d^2-1} a_{jk} P_k, 
\label{eq:Aj}
\eeq
where $a_{jk}\in\mathbb{R}$. This amounts to control noise that perturbs the intended gate Hamiltonian $P_j$ by a bounded (but not necessarily local) operator. 
We may now write the noisy version of the gate as:
\beq
U'(\theta_\mu)=\exp(-i\theta_\mu(P_{j(\mu)}+ A_{j(\mu)})/2) ,
\label{eq:noisyU}
\eeq
where the prime indicates the presence of noise.
Note that this noise model includes both under/over-rotation and axis-angle errors. In the former case $a_{jk}= a\delta_{jk}$, so that $\theta_\mu \mapsto \theta_\mu +  a$, while in the latter case $a_{jk}\ne \delta_{jk}$, so that the rotation axis is no longer $P_j$. Since $j=j(\mu)$, this noise model also accounts for the location of the errors in the circuit. The errors can be deterministic or stochastic, but our model assumes that they are constant throughout the duration of each gate. 

Using \cref{eq:21c} and \cref{eq:cost-diff}, the noisy version of the cost function and its gradient with respect to $\theta_\mu$ can be written as:
\bes
    \begin{align}
    C'(\vect{\theta}) &= \Tr[\tilde{H}  \mcU'(\theta_\mu)(\tilde{\rho})] \\
        \frac{\partial C'(\vect{\theta})}{\partial \theta_\mu}&=-\frac{i}{2}\Tr[ \tilde{H}\mcU'(\theta_\mu)([P_{j(\mu)}+ A_{j(\mu)},\tilde{\rho}])].
    \label{eq:noisy-C-diff}
    \end{align}
\ees
Expanding \cref{eq:noisy-C-diff} using \cref{eq:pauli-identity}:
\bes
\label{eq:dC'indetail}
    \begin{align}
    &\frac{\partial C'(\vect{\theta})}{\partial \theta_\mu}=-\frac{i}{2}\Tr[ \tilde{H}\mcU'(\theta_\mu)([P_j+ A_j,\tilde{\rho}])]\\
        &=-\frac{i}{2}\Tr[ \tilde{H}\mcU'(\theta_\mu)([P_j,\tilde{\rho}])]\notag \\
        &\quad-\frac{i}{2}\sum_ka_{jk}\Tr[ \tilde{H}\mcU'(\theta_\mu)([P_k,\tilde{\rho}])]\\
        \notag
        &=\frac{1}{2}\Tr[ \tilde{H}\mcU'(\theta_\mu)\left(\mcU_j\left(\frac{\pi}{2}\right)(\tilde{\rho}) -\mcU_j\left(-\frac{\pi}{2}\right)(\tilde{\rho}) \right)]\\
        &\quad+\frac{1}{2}\sum_ka_{jk}\Tr[ \tilde{H}\mcU'(\theta_\mu)\times\\
        &\qquad\qquad\qquad\quad\left(\mcU_k\left(\frac{\pi}{2}\right)(\tilde{\rho}) -\mcU_k\left(-\frac{\pi}{2}\right)(\tilde{\rho}) \right)], \notag
    \end{align}
\ees
where, in the last two lines, we have used $\mcU_j$ to denote a map generated by $P_j$, so that it is distinguished from $\mcU'$, which is generated by $P_j+A_j$ as given in \cref{eq:noisyU}.
We also denote
\beq
\tilde{\rho}_{j(\mu)}^{\pm}\equiv \mcU'(\theta_\mu)\mcU_{j(\mu)}\left(\pm\frac{\pi}{2}\right)(\tilde{\rho}) = \frac{1}{\sqrt{d}}F_0+\tilde{\vect{v}}_{j(\mu)}^\pm\cdot\vect{F}, 
\label{eq:rhohatexpansion}
\eeq
where $\tilde{\vect{v}}_{j(\mu)}^\pm$ is the corresponding coherence vector after an expansion in a nice operator basis. This bifurcation into the two paths labeled $\pm$ will turn out to be key to the NIBP phenomenon.

Let 
\bes
\begin{align}
\label{eq:w-hat-def}
\tilde{\vect{w}}_{j(\mu)} &= \tilde{\vect{v}}^+_{j(\mu)}-\tilde{\vect{v}}^-_{j(\mu)} ,\\ 
\tilde{\xi}_{j(\mu)} &=\tilde{\rho}_{j(\mu)}^+-\tilde{\rho}_{j(\mu)}^- = \tilde{\vect{w}}_{j(\mu)}\cdot\vect{F}.  
\end{align}
\ees
Substituting \cref{eq:rhohatexpansion} into \cref{eq:dC'indetail}, we have:
\bes
    \begin{align}
    &\frac{\partial C'(\vect{\theta})}{\partial \theta_{\mu}}\\ \notag
    &=\frac{1}{2}\Tr[ \tilde{H}(\tilde{\rho}_{j}^+-\tilde{\rho}_{j}^-)]+\frac{1}{2}\sum_ka_{j k}\Tr[\tilde{H}(\tilde{\rho}_k^+-\tilde{\rho}_k^-)]\\
        &=\frac{1}{2}\Tr(\tilde{H}\tilde{\xi}_j)+\frac{1}{2}\sum_k a_{j k}\Tr(\tilde{H}\tilde{\xi}_k)\\
        &=\frac{1}{2} \Tr(\tilde{H} \tilde{\vect{w}}_j \cdot\vect{F})+\frac{1}{2}\sum_ka_{j k}\Tr(\tilde{H}\tilde{\vect{w}}_k\cdot\vect{F}) \\
        &=\frac{1}{2}\sum_{l=1}^{d^2-1} (\tilde{\vect{w}}_j)_l \tilde{h}_l+\frac{1}{2}\sum_{k}a_{j k}\sum_{l=1}^{d^2-1} (\tilde{\vect{w}}_k)_l \tilde{h}_l \\
        &=\frac{1}{2} \tilde{\vect{w}}_{j(\mu)} \cdot \tilde{\vect{h}}+\frac{1}{2}\sum_{k}a_{j k}\tilde{\vect{w}}_{k(\mu)} \cdot \tilde{\vect{h}} ,
    \end{align}
\ees
where we denoted $\tilde{h}_l \equiv \Tr(\tilde{H} F_l)$ and selectively added the $\mu$-dependence for emphasis (though both $j$ and $k$ depend on $\mu$). Since $\tilde{H}$ is related to $H$ via a unitary transformation [\cref{eq:tildeH}] they have the same norm, and likewise $\| \tilde{\vect{h}} \| = \| \vect{h} \|$. We can now bound the derivative as follows:
\begin{align}
        \left|\frac{\partial C'(\vect{\theta})}{\partial \theta_{\mu}}\right|
\label{eq:control-N-bound}
        &\le \frac{1}{2} \|\tilde{\vect{w}}_{j(\mu)} \|\|\vect{h}\|+\frac{1}{2}\sum_{k}|a_{j k}|\|\tilde{\vect{w}}_{k(\mu)} \| \| \vect{h} \|.
\end{align}
The noise term in \cref{eq:control-N-bound} (the second term) makes this bound looser than the noiseless case (just the first term), and it might seem that this could prevent the gradient from becoming vanishingly small. However, as long as $\|\tilde{\vect{w}}_{k(\mu)}\|$ is exponentially suppressed, it is not loose enough to escape the NIBP. We discuss this in more detail in \cref{sec:NIBP}.

\subsection{Random unitary noise}

The noise model in \cref{sec:control-noise} is semiclassical, in that the bath is treated classically; a fully quantum noise model would replace \cref{eq:Aj} by $\sum_{k=0}^{d^2-1} a_{jk} P_k\otimes B_k + H_B$, where $\{B_k\}$ and $H_B$ are, respectively, bath operators and the bath Hamiltonian. This would change the system-only unitary $U'(\theta_\mu)$ into a system-bath unitary, but we do not consider this case here. Instead, to account for a quantum noise model of faulty gates, we consider a (unital) noise map defined by a set of unitary operators and their corresponding probabilities $\{p_k,U_k\}$. I.e., with probability $p_k$ the unitary $U_k=\exp(-i\theta_{\mu,k} P_k/2)$ is applied. Only one of these unitaries is the intended one; w.l.o.g., we call this index $j=j(\mu)$. We assume that $\sum_{k\neq j}p_k\ll p_j$ (i.e., $p_j\lesssim 1$).

For simplicity, we may assume that $\theta_{\mu,k}=\theta_\mu$, since the case $\theta_{\mu,k}=\theta_\mu+(\Delta\theta)_k$ was already accounted for in \cref{sec:control-noise}. Hence, each $U_k$ is implicitly $U_k(\theta_\mu)$, which has the same angle dependence as $U_l(\theta_\mu)$ for $k\ne l$.

This noise model can be written in the Kraus representation as 
\bes
\begin{align}
\mcV(X) &= \sum_k p_k U_k X U_k^\dag\\
&=p_j\mcU(X)+\sum_{k\neq j}p_k\mcV'_k(X),
\label{eq:randKOSR}
\end{align}
\ees
where $\mcV'_k(X)= U_k X U_k^\dag$ for $k\neq j$ and $\{\sqrt{p_k}U_k\}$ are the Kraus operators. 

Using a similar approach as in the derivation of \cref{eq:control-N-bound}, we find:
\begin{align}
\label{eq:random-N-bound}
    \left|\frac{\partial C'(\vect{\theta})}{\partial \theta_\mu}\right|
    &\le p_{j(\mu)}\left|\frac{\partial C(\vect{\theta})}{\partial \theta_\mu}\right| \\
    &\qquad\qquad+\frac{1}{2}\sum_{k\neq j}p_{k(\mu)} \|\tilde{\vect{w}}_{k(\mu)}\|\| \vect{h}\|.
    \notag
\end{align}
Details are given in \cref{app:proof_of_random-N-bound}.
The conclusion regarding this bound is similar to the case discussed in \cref{sec:control-noise}.

\subsection{Bounds on $\|\vect{h}\|$}
\label{ss:bound-H}

Both \cref{eq:control-N-bound,eq:random-N-bound} involve the Hamiltonian norm $\|\vect{h}\|$, so we next bound this quantity.

Consider the scaling of the HS norm of the problem Hamiltonian. Writing such Hamiltonians as $H = \sum_{j=0}^{d^2-1} h_j F_j$, we have $\| H\|_2^2 = \sum_{j=0}^{d^2-1} h_j^2 = h_0^2+\|\vect{h}\|^2$. In practice, we choose the nice operator basis $\{F_j\}$ as the normalized Pauli basis, $\{P_j\}/\sqrt{d}$.

Assume that the Pauli strings are ordered by Hamming weight $k$.
We say that $H$ is $K$-local if $h_j=0$ for $j>K$ with $K$ a constant independent of $n$. The number of non-zero $h_j$ terms in $\vect{h}$ is at most $\sum_{k=1}^{K} \binom{n}{k}$, and a crude upper bound for $K\le n/2$ (which holds in the $K$-local case) is 
\bes
\begin{align}
\sum_{k=1}^{K} \binom{n}{k} &\le K\max_{k\in [0,K]} \binom{n}{k} = K \binom{n}{K}\\
& = K\frac{n(n-1)\cdots(n-K+1)}{K!} \\
& \le \frac{n^K}{(K-1)!} .
\end{align}
\ees
Thus, 
\beq
\|\vect{h}\| \le  \frac{h}{\sqrt{(K-1)!}} n^{K/2} ,\quad h\equiv\max_{j> 0} h_j .
\label{eq:h-vec-bound}
\eeq

\section{Cost Function Concentration and Noise-Induced Limit Sets}
\label{sec:NILS}

Under the unital map setting of generalized Pauli noise, Ref.~\cite{Wang2021} has shown that the cost function concentrates. We now extend this to the action of HS-contractive non-unital noise maps and show furthermore that a new phenomenon of noise-induced limit sets arises in this context.

\subsection{Cost function concentration under non-unital noise}

Expanding $H$ in a nice operator basis, we have: 
\beq
H = \sum_{j=0}^{d^2-1} h_j F_j = h_0 F_0+ \vect{h}\cdot \vect{F} ,
\label{eq:H-expand}
\eeq 
so that $\vect{h}$ collects the coordinates of the traceless component of $H$. The cost function can be written as:
\begin{align}
\label{eq:concentrate-C-1}
C(\vect{\theta}) &=\Tr[H \rho(\vect{\theta}) ] =\frac{1}{d}\Tr(H)+\vect{v}\cdot \vect{h} .
\end{align}

We already showed that when the noise is part of the gate, the gradient of the noisy cost function is bounded as in \cref{eq:control-N-bound,eq:random-N-bound}. We now consider noise between gate applications, which we model as a concatenation of non-unitary CPTP maps $\mathcal{N}_l$ after applying all the noisy gates in the $l$'th layer: 
\beq
\rho_{l+1} = [\mathcal{N}_{l+1} \circ \mathcal{U}'(\vect{\theta}_{l+1})](\r_l)\quad \forall l\ge 0 .
\label{eq:rho_l}
\eeq
I.e., for the evolution of a circuit of depth $L$, with the initial state $\rho_0$, we have
\beq
\rho(\vect{\theta})\equiv \rho_L(\vect{\theta}) = [\mathcal{N}_L \circ \mathcal{U}'(\vect{\theta}_L)\circ \cdots \circ \mathcal{N}_1 \circ \mathcal{U}'(\vect{\theta}_1)] (\rho_0) ,
\label{eq:fullcircuit}
\eeq
where, as before $\vect{\theta} = \{\vect{\theta}_l\}_{l=1}^L$, and $\mathcal{U}'(\vect{\theta}_l)$ denotes the noisy unitary superoperator in the $l$'th layer, formed from gates of the form given in \cref{eq:noisyU}.

Let $\vect{v}_{l}$ be the coherence vector corresponding to $\r_l$ in \cref{eq:rho_l}.
Let us denote the transformed coherence vector after $\mcN_l\circ\mcU'_l(\vect{\theta}_l)$ by 
\beq
\vect{v}_{l}=\O_l\vect{v}_{l-1}+\vect{c}_l\ , \quad \O_l \equiv M_l O_l ,
\label{eq:v_l}
\eeq
with $O_l$ the orthogonal rotation corresponding to $\mcU'_l(\vect{\theta}_l)$, and $M_l$ the rotation+dilation and $\vect{c}_l$ the affine shift corresponding to $\mcN_l$ (see \cref{sec:non-unital}). Thus, $O_l$ is norm-preserving
and $\vect{c}_l$ is either zero or non-zero, depending on whether $\mcN_l$ is unital or non-unital, respectively (\cref{lemma2} and \cref{lemma3}).
The dilation part of $M_l$ contracts the vector it acts on. Let $q_l$ be the contractivity factor associated with $\O_l$, i.e., for any vector $\vect{v}$,
\beq
\|\O_l\vect{v}\| = q_l \|\vect{v}\|, \quad 0\le q_l < 1 .
\label{eq:q_l}
\eeq

Expanding the recursion given by \cref{eq:v_l}, with the initial condition $\vect{v}_0$ corresponding to $\r_0$, we obtain, with $\vect{d}_1\equiv \vect{c}_1$ and $l\le L$:
\bes
\label{eq:v_j<l}
\begin{align}
\label{eq:v_j<l-1}
\vect{v}_j &= \O_j \cdots  \O_1 \vect{v}_0 + \vect{d}_j \ ,  \qquad  \quad 1\le j\le L\ \\
\vect{d}_j &=\O_j \cdots \O_2 \vect{c}_1 +\O_j \cdots \O_3 \vect{c}_2 + \cdots \notag\\
&\qquad\qquad\qquad\qquad\qquad+ \O_j \vect{c}_{j-1}+\vect{c}_j \\
&= \sum_{r=1}^{j-1}\left( \prod_{s=j}^{r+1} \O_s \right) \vect{c}_r+\vect{c}_j, \quad  1\le j\le L.
\label{eq:v_j<l-2}
\end{align}
\ees

Note that $\vect{d}_j$ is entirely a property of the maps and does not depend on the system state $\vect{v}_j$. In other words, $\vect{d}_j$ contains no useful information about the state of the computation carried out by the circuit. Using \cref{eq:v_j<l}, we can write a transformed $\vect{v}$ at the end of the circuit as
\begin{align}
\label{eq:v_L}
\vect{v}_L = \O_L \cdots  \O_1 \vect{v}_0 + \vect{d}_L ,
\end{align}
where $\vect{d}_L$ arises from the noise map combined with the unitary VQA map. 
Substituting this into \cref{eq:concentrate-C-1} with $\vect{v}\equiv \vect{v}_L$, we have:
\begin{align}
\label{eq:C-1dTrH}
C(\vect{\theta})
= \frac{1}{d}\Tr(H) + \O_L \cdots  \O_1 \vect{v}_0\cdot \vect{h} + \vect{d}_L \cdot \vect{h} .
\end{align}
Grouping together the terms that do not get contracted over the layers of the VQA circuit, we find:
\bes
\label{eq:concentrate-C}
\begin{align}
|C(\vect{\theta})-\frac{1}{d}\Tr(H)-\vect{d}_L \cdot \vect{h}|&= q^L|\vect{v}_0\cdot \vect{h}|\\
&\le q^L\|\vect{v}_0\|\| \vect{h}\|\\
\label{eq:concentrate-C-3}
&\le q^L\| \vect{h}\|,
\end{align}
\ees
where
\beq
q^L \equiv q_1\cdots q_L\ , \quad 0\le q_l < 1\ \forall l .
\eeq 
Here $q_l$ is the contractivity factor associated with $\O_l$, so that the effective contractivity factor $q<1$. Moreover, recall from \cref{eq:h-vec-bound} that $\|\vect{h}\| \le \frac{h}{\sqrt{(K-1)!}} n^{K/2}$, where $K$ (a constant) is the locality of $H$. We have thus proven that the cost function landscape concentrates:
\begin{mytheorem}
\label{th:concentration}
The cost function of HS-contractive non-unital maps concentrates for any VQA circuit with greater than logarithmic depth, i.e., if $L \in \omega[\log(n)]$ then\footnote{Recall the little omega notation; $f(x)\in\omega(g(x))$ means that for any positive constant $c$, there exist a real constant $x_0$ such that $f(x) > c g(x)$ $\forall x\ge x_0$.}
\bes
\begin{align}
\label{eq:concentrate-C2}
|C(\vect{\theta})-C_L|&\le \frac{h}{\sqrt{(K-1)!}} n^{K/2} q^L \\
& \to 0 \quad \text{as}\quad L\to\infty \\
C_L \equiv& \frac{1}{d}\Tr(H) + \vect{d}_L \cdot \vect{h}
\label{eq:C_NILS}
\end{align}
\ees
\end{mytheorem}
The meaning of the $L \in \omega[\log(n)]$ condition is that this result holds as long as $L$ is large enough, i.e., except for circuits that are shallower than logarithmic depth, since then the $\|\vect{h}\|$ factor can counteract the $q^L$ factor in \cref{eq:concentrate-C-3}.

\subsection{Noise-induced limit set}
\label{sec:NILS-subsec}

While \cref{th:concentration} shows that the cost function concentrates on $C_\infty = \frac{1}{d}\Tr(H) + \vect{d}_\infty \cdot \vect{h}$, this does not mean that it tends to a fixed point. The reason is that the vector $\vect{d}_L$ is affected by the unitary transformations in the VQA circuit, which rotate it, and by noise in the circuit, which rotates and contracts it. The rotations prevent convergence on a fixed point, and instead $C_\infty$ lies in an interval. We call this interval the \emph{noise-induced limit set} (NILS). Our next result characterizes the NILS.

We can rewrite \cref{eq:v_j<l-2} as:
\bes
    \begin{align}
        \vect{d}_L&=\O_L\cdots\O_2\vect{c}_1+\O_L\cdots\O_3\vect{c}_2+\cdots\\ \notag
        &\qquad+\O_L\vect{c}_{L-1}+\vect{c}_{L}\\
 \label{eq:d_L-again}
       &= p_{1}^{L-1}\Theta_1\vect{c}_1+p_{2}^{L-2}\Theta_2\vect{c}_2+\cdots\\ \notag
       &\qquad+p_{L-1}\Theta_{L-1}\vect{c}_{L-1}+\vect{c}_{L},
    \end{align}
\ees
where $\Pi_{l=i+1}^{L}\O_{l}\equiv p_i^{L-i}\Theta_i$, i.e., combining the transformations of $\vect{c}_i$ into an overall rotation $\Theta_i$ and an overall contraction $p_i^{L-i}$, where $0\le p_i < 1, \forall i$. Let $p \equiv \max_i p_i$, a noise contraction factor. Then:

\begin{myproposition}
\beq
\label{eq:d.h-bound}
|\vect{d}_L \cdot \vect{h}| \le \frac{1-p^L}{1-p}\frac{\|\vect{h}\|}{\sqrt{1-1/d}} \equiv \Lambda_L \le \Lambda_\infty .
\eeq
\end{myproposition}
The proof is given in \cref{app:NILS}.

Taking the limit $L\to\infty$, we can identify the \emph{noise-induced limit set} (NILS) that the cost function concentrates in: 
\begin{mycorollary}
\label{cor:NILS}
\beq
C_\infty \in [\frac{1}{d}\Tr(H) - \Lambda_\infty , \frac{1}{d}\Tr(H) + \Lambda_\infty] \equiv \mathrm{NILS}.
\label{eq:C_NILS_infty}
\eeq 
where 
\beq
\Lambda_\infty = \frac{1}{1-p}\frac{\|\vect{h}\|}{\sqrt{1-1/d}}  .
\eeq
\end{mycorollary}

\begin{proof}
It follows from \cref{eq:d.h-bound} that $-\Lambda_\infty \le \vect{d}_L \cdot \vect{h} \le \Lambda_\infty$; adding $\frac{1}{d}\Tr(H)$ to all sides gives $-\Lambda_\infty + \frac{1}{d}\Tr(H) \le C_L \le \Lambda_\infty + \frac{1}{d}\Tr(H)$ $\forall L$.
\end{proof}

Thus, the NILS is an interval dictated by the dimension of the system, the problem Hamiltonian, and the noise through the contraction factor $p$. Note again, that $C_\infty$ and $\vect{d}_\infty$ are notations we have chosen and are not meant to suggest that a limit should exist as a point. As we have shown, the limit is an \emph{interval} that we call a noise-induced limit set.

\subsection{The unital case}

Since the unital case is the case for which $\vect{d}_L=\vect{0}$ [all the shift vectors $\vect{c}_j$ vanish in \cref{eq:v_j<l-2}], we have: 
\begin{mycorollary}
The cost function of unital maps concentrates super-polynomially for any VQA circuit with logarithmic depth, i.e., if $L = \omega[\log(n)]$, and exponentially for any VQA circuit with at least linear depth, i.e., if $L = \Omega(n)$.
\end{mycorollary}

The latter statement recovers the concentration result for unital maps of Ref.~\cite[Lemma 1]{Wang2021}, which holds ``whenever the number of layers $L$ scales linearly with the number of qubits'', i.e., if $L = \omega(n)$ in our notation.

Note that the NILS in the case of HS-contractive non-unital noise is worse for VQA circuit performance than unital noise since knowledge of $C_L$ requires a precise characterization of each of the intermediate non-unital maps (to determine $\vect{d}_L$), whereas in the unital case, the NILS is determined purely by the target Hamiltonian $H$, and becomes a fixed point: 
\beq
C^{\text{unital}}_\infty = \frac{1}{d}\Tr(H)  = \<H\>_{I/d} ,
\label{eq:C_NILS-unital}
\eeq
i.e., the expectation value of the Hamiltonian with respect to the fully mixed state.

Next, we discuss the NIBP phenomenon for unital and HS-contractive non-unital maps and show that it appears for the former (in agreement with Ref.~\cite{Wang2021}) but not for the latter.

\section{NIBP via parameter shift rules}
\label{sec:NIBP}

From now on, when we use $C(\vect{\theta})$, we refer to the cost function in the presence of a noise map. In this section, we bound $|\partial C(\vect{\theta})/\partial \theta_{\mu}|$, the magnitude of the cost function gradient, and show that it is exponentially small in the circuit depth $L$ for both unital and non-unital noise, for $n \geq 1$ qubits. We interchangeably use both $lm$ and $\mu$ to denote the gate location.

Our proof in the unital case is simpler and more general than that of Ref.~\cite{Wang2021}, but the HS-contractive non-unital case is our main new result. The bound we find has implications for the many applications of VQA where the goal is to learn the optimal parameters, i.e., when one is primarily concerned with trainability, and hence the gradient is a key quantity of interest.

At this point, it is useful to provide a formal definition of noise-induced barren plateaus. The following definition is inspired by Ref.~\cite[Theorem~1]{Wang2021}:

\begin{mydefinition}
\label{def:NIBP}
A cost function $C(\vect{\theta})$ exhibits a noise-induced barren plateau (NIBP) if the magnitude of its gradient, $\left| \frac{\partial C(\vect{\theta})}{\partial \theta_{\mu}}\right|$, decays exponentially as a function of the circuit depth $L$ for all $L$ larger than some constant $L_0\ge 1$, independently of $l$ and $m$, even for constant-width circuits.
\end{mydefinition}
Thus, NIBPs flatten the entire control landscape independently of the location $\mu=(l,m)$ of the gate in the circuit at which the derivative is taken. Moreover, we impose the condition that the result holds even for constant-width circuits in order to preclude a measure concentration-type argument, which is typical of noise-free barren plateaus~\cite{McClean2018}. In this sense, NIBPs are distinct from the latter, for which the global minimum can be embedded inside a deep, narrow valley in the control landscape~\cite{Cerezo2021-bp}.

Using the PSR [\cref{eq:psr}], and choosing the target operator to be a Hamiltonian $H$, we have:
\begin{align}
         \frac{\partial C(\vect{\theta})}{\partial \theta_{\mu}}
        =\frac{1}{2}\Tr[H(\rho(\vect{\theta}_{\mu}^+)-\rho(\vect{\theta}_{\mu}^-))] ,
\end{align}
where $\vect{\theta}_{\mu}^\pm \equiv \vect{\theta}\pm\vect{\theta}_{\mu}^{\pi/2}$. We use \cref{eq:rho-nice} to write 
\beq
\label{eq:v_pm-def}
\rho(\vect{\theta}_{\mu}^\pm)=\frac{1}{d}I + \sum_{i=1}^{d^2-1}({v}^\pm_{\mu})_i F_i = \frac{1}{\sqrt{d}}F_0+\vect{v}^\pm_{\mu}\cdot\vect{F}\ ,
\eeq
so that
\beq
\rho(\vect{\theta}_{\mu}^\pm) = \frac{1}{\sqrt{d}}F_0+\vect{v}_{\mu}^\pm\cdot\vect{F} .
\eeq 
Let 
\beq
\label{eq:v_mu^L}
\tilde{\vect{v}}_{\mu}^L \equiv \vect{v}^+_{\mu}-\vect{v}^-_{\mu},
\eeq 
where we added the $L$ subscript as a reminder that $\tilde{\vect{v}}_{\mu}^L$ corresponds to the difference between two states obtained at the end of the circuit. Using \cref{eq:H-expand}:
\bes
\begin{align}
         \frac{\partial C(\vect{\theta})}{\partial \theta_{\mu}}
        &=\frac{1}{2} \Tr(H \tilde{\vect{v}}_{\mu}^L\cdot\vect{F}) \\
        & = \frac{1}{2}\sum_{i=1}^{d^2-1} (\tilde{\vect{v}}_{\mu}^L)_i \Tr(H F_i) \\
        &=\frac{1}{2}\sum_{i=1}^{d^2-1} (\tilde{\vect{v}}_{\mu}^L)_i h_i .
\end{align}
\label{eq:diff-C-expand}
\ees

We thus arrive at the key result that the magnitude of the cost function gradient can be written simply as:
\begin{align}
\label{eq:gradient-result}
         \left| \frac{\partial C(\vect{\theta})}{\partial \theta_{\mu}}\right|
        =\frac{1}{2}  \left|\tilde{\vect{v}}_{\mu}^L \cdot \vect{h}\right| ,
\end{align}
which expresses the cost function gradient in terms of the overlap of the difference between two coherence vectors with the coordinates of the traceless component of the target Hamiltonian.

\subsection{NIBP in the unital case}
\label{sec:unital-bound}

We now prove the existence of an NIBP for unital maps, in agreement with Ref.~\cite{Wang2021}.

\begin{mytheorem}
\label{th:unital-NIBP}
Assume that the maps $\{\mcN_l\}_{l=1}^L$ in the VQA circuits described by  \cref{eq:fullcircuit} are all unital but non-unitary and that the Hamiltonians generating the gates in the circuit and the control noise are $K$-local. Let $n$ denote the circuit width. Assume, moreover, that 
either $L = c[\ln(n)]^{Q}$ with $Q>1$, or $Q=1$ and $K<2c\ln(1/r)$, where $c>0$ is a constant and $0<r<1$ is a contractivity factor associated with the maps maps $\{\mcN_l\}_{l=1}^L$ [defined in \cref{eq:r-def}].
The cost function of such circuits exhibits an NIBP.
\end{mytheorem}

\begin{proof}
Using 
the Cauchy-Schwarz inequality, \cref{eq:gradient-result} yields:
\begin{equation}
    \left|\frac{\partial C(\vect{\theta})}{\partial \theta_{\mu}}\right|
        \leq \frac{1}{2} \|\tilde{\vect{v}}_{\mu}^L\| \|\vect{h}\| ,
    \label{eq:C-upper-bound}
\end{equation}

Since the maps $\mcN_l$ in \cref{eq:fullcircuit} are all unital, $\mcN_l\circ\mcU'_l$ is also unital since $\mcU'_l$ is unitary and hence unital. Let $\vect{v}_{l}$ be the coherence vector corresponding to $\r_l$ in \cref{eq:rho_l}.
It follows from \cref{lemma2} that $\|\vect{v}_{l+1}\|=r_l \|\vect{v}_l\|$, where $0\le r_l < 1$ $\forall l$ is the contractivity factor associated with $\mcN_l$.
Since $\vect{v}^{\pm}_{\mu}$ are the coherence vectors corresponding to $\rho(\vect{\theta}_{\mu}^{\pm})$ -- the states obtained at the end of the circuit but with $\vect{\theta}$ shifted by $\vect{\theta}_{\mu}^\pm$ -- the same reasoning applies, and we can write:
    \begin{align}
    \label{eq:psr-con}
        \|\vect{v}^+_{\mu}\|=p^L \|\vect{v}_0\|\ , \quad 
        \|\vect{v}^-_{\mu}\|=q^L \|\vect{v}_0\| ,
    \end{align}
where $\vect{v}_0$ is the coherence vector of $\r_0$ (the initial density matrix), and
\beq
p^L \equiv p_1\cdots p_L\ ,q^L \equiv q_1\cdots q_L\ ,  0\le p_l,q_l < 1\ \forall l.
\eeq 
Here $p$ and $q$ are the effective contractivity factors associated with the two paths involving \cref{eq:fullcircuit} and $\vect{\theta}_{\mu}^+$ or $\vect{\theta}_{\mu}^-$, respectively. 
Using the elementary inequality $\|\vect{a}-\vect{b}\| \le \|\vect{a}\| + \|\vect{b}\|$ we now have
\bes
\begin{align}
\label{eq:norm-vec-unital}
\|\tilde{\vect{v}}_{\mu}^L\|&= \| \vect{v}^+_{\mu}-\vect{v}^-_{\mu}\| \le \|\vect{v}_0\| (p^L+q^L) \le  2 r^{L} \\
\label{eq:r-def}
r&\equiv\max(q,p),
\end{align}
\ees
where we used $\|\vect{v}_0\| < 1$ [\cref{eq:v-bounded}]. Thus, using \cref{eq:C-upper-bound}:
    \begin{align}
    \label{eq:NIBP1}        
    \left|\frac{\partial C(\vect{\theta})}{\partial \theta_{\mu}}\right|
        \leq \|\vect{h}\| r^{L}\ ,\quad \ 0\leq r< 1 .
    \end{align}

The noise channel in \cref{eq:fullcircuit} already includes the situations with ideal gates [$\mcU$, \cref{eq:idealU}] and control noise in the gates [$\mcU'$, \cref{eq:noisyU}]. To make this explicit, we can write:
\begin{align}
    \label{eq:control-N}      
    \left|\frac{\partial C'(\vect{\theta})}{\partial \theta_{\mu}}\right|_\text{ctrl}
        \leq \|\vect{h}\| r_\text{c}^{L}\ ,\quad \ 0\leq r_\text{c}< 1 .
    \end{align}

We also model the random unitary noise which is a linear combination of unitary noise channels [\cref{eq:randKOSR}]. Each summand can be modeled by \cref{eq:fullcircuit}, hence, is bounded as in \cref{eq:NIBP1}:
\bes
\label{eq:random-N}
\begin{align}
    \left|\frac{\partial C'(\vect{\theta})}{\partial \theta_{\mu}}\right|_\text{rand}
    &\le \sum_{k'}p'_{k'}\|\vect{h}\|r_\text{r}^L,\quad \ 0\leq r_\text{r}< 1,
\end{align}
\ees
which is still bounded since $k$ does not scale exponentially with $L$.

Recalling \cref{eq:h-vec-bound}, the magnitude of the gradient of the cost function [whether \cref{eq:NIBP1}, \cref{eq:control-N}, or \cref{eq:random-N}] satisfies the bound
\beq
\left| \frac{\partial C(\vect{\theta})}{\partial \theta_{\mu}}\right|,\left|\frac{\partial C'(\vect{\theta})}{\partial \theta_{\mu}}\right|_\text{ctrl},\left|\frac{\partial C'(\vect{\theta})}{\partial \theta_{\mu}}\right|_\text{rand} \le
                  g n^{K/2} r^{L} ,
\eeq
where $g=g'g''$, and $g'=\frac{h}{\sqrt{(K-1)!}}$ and $g''$ are positive constants independent of $n$.

If $L = c[\ln(n)]^{Q}$ ($c>0$), i.e., the circuit has sublogarithmic ($0<Q<1$), logarithmic ($Q=1$), or superlogarithmic ($Q>1$) depth, then $n = \exp[(L/c)^{1/Q}] $, so that 
\beq
g n^{K/2} r^{L} = O\left( \exp[\frac{1}{2}K(L/c)^{1/Q} - \ln(1/r) L]\right) .
\label{eq:log-decay}
\eeq 
This quantity decays exponentially provided $\frac{1}{2}K(L/c)^{1/Q} < \ln(1/r) L$. Solving this inequality for $L$, we see that for any $Q>1$, this is again exponentially suppressed in the circuit depth $L$ for all $L>L_0$, where 
\beq
L_0 =  c^{1-Q}\left(\frac{K/2}{\ln(1/r)}\right)^{\frac{Q}{Q-1}} .
\eeq
In both cases the cost function gradient decays exponentially with the circuit depth $L$, i.e., we have an NIBP as per \cref{def:NIBP}. The circuit under extra control noise and random unitary noise could not escape NIBP if the original noisy circuit exhibits NIBP.

When $Q=1$ (logarithmic circuit depth), the r.h.s. of \cref{eq:log-decay} decays exponentially in $L$ if $K<2c\ln(1/r)$. This can be interpreted as an upper bound on the locality of the Hamiltonian and the control noise in terms of the largest contractivity factor ($r$) of the unital noise maps in the circuit. If this condition is satisfied then we again find an NIBP.
\end{proof}

Note that if $K\ge 2c\ln(1/r)$, or if $Q<1$ (sublogarithmic circuit depth), then we cannot conclude from our bounds that the circuit exhibits an NIBP. In other words, an NIBP may still occur but this cannot be inferred from our analysis.

\subsection{The non-unital case}
\label{sec:nonunital-bound}

Now assume that at least one of the maps $\mcN_l$ in \cref{eq:fullcircuit} is non-unital. We will show that in contrast to the unital case, in the non-unital case $ \left| \frac{\partial C(\vect{\theta})}{\partial \theta_{lm}}\right|$ can be lower-bounded by a quantity that is non-vanishing even for arbitrarily deep circuits. This means that there is no guarantee of an NIBP in the non-unital case.

Recall \cref{eq:v_l}. The effect of shifting $\vect{\theta}$ by $\vect{\theta}_{lm}^\pm$ at a single location $\mu =(l,m)$ in the circuit is that in layer $l$, and only in this layer, we have two different unitaries $\mcU'_l(\vect{\theta}^{\pm}_{lm})$, and correspondingly two different orthogonal rotations $O_{lm}^{\pm}$: 
\beq
\vect{v}^{\pm,l}_{lm}=\O_{lm}^\pm \vect{v}_{l-1} +\vect{c}_l  \ , \quad \O_{lm}^\pm\equiv M_{l}O_{lm}^\pm .
\label{eq:v_l^pm}
\eeq
Note that prior to this location, i.e., for all $l'<l$, the bifurcation into the two paths labeled $\pm$ has not yet happened, which is why $\vect{v}_{l-1}$ does not carry a $\pm$ label. 

\subsubsection{No guarantee of an NIBP: an example}

As a simple example that demonstrates why there is no guarantee of an NIBP in the non-unital case,  
assume that all $\mcN_l$ are unital except for the last two, i.e., $\mcN_l$ is non-unital only for $l=L-1$ and $l=L$. Writing the last two coherence vectors explicitly then gives:
\bes
\begin{align}
\label{eq:51a}
\vect{v}_{Lm}^{\pm,L} &= \O_{Lm}^{\pm}\vect{v}_{L-1} + \vect{c}_L\\
\label{eq:51b}
\vect{v}_{L-1} &= \O_{L-1}\vect{v}_{L-2} + \vect{c}_{L-1} ,
\end{align}
\ees
and substituting \cref{eq:51b} into \cref{eq:51a} yields:
\beq
\vect{v}_{Lm}^{\pm,L} = \O_{Lm}^{\pm}\O_{L-1}\vect{v}_{L-2} +  \O_{Lm}^{\pm}\vect{c}_{L-1} + \vect{c}_{L} .
\eeq
The term $\O_{Lm}^{\pm}\O_{L-1}\vect{v}_{L-2}$ is identical to the terms that appear in the unital case, so we know from the proof of \cref{th:unital-NIBP} that its norm is $O\left(e^{-L}\right)$; therefore, for simplicity, let us neglect it entirely.
Subtracting then yields:
\bes
\label{eq:vL-lower}
\begin{align}
\tilde{\vect{v}}^L_{Lm} =  \vect{v}_{Lm}^{+,L} -\vect{v}_{Lm}^{-,L}  = M_L(O_{Lm}^{+} - O_{Lm}^{-})\vect{c}_{L-1} .
\end{align}
\ees
We know from \cref{lemma3} that $0<\|\vect{c}_{L-1}\| < 1$. Thus, the vector $(O_{Lm}^{+} - O_{Lm}^{-})\vect{c}_{L-1}$ has a non-zero $L$-independent norm determined by the two different rotations $O_{Lm}^{\pm}$. Applying $M_L$ to it can rescale its norm by, at most, a constant ($L$-independent) factor. Thus, the argument leading to the exponentially small upper bound on $\|\tilde{\vect{v}}^L_{Lm}\|$ in \cref{eq:norm-vec-unital} does not hold in this case. Instead, we now have, from \cref{eq:gradient-result}:
\begin{align}
         \left| \frac{\partial C(\vect{\theta})}{\partial \theta_{Lm}}\right|
        =\frac{1}{2} \left| \tilde{\vect{v}}^L_{Lm} \cdot \vect{h}\right| .
        \label{eq:non-unital-cost}
\end{align}
It follows from standard Levy's lemma-type arguments that two randomly chosen unit vectors in $\mathbb{R}^D$ have overlap $\sim 1/\sqrt{D}$ (see, e.g., Ref.~\cite{random-vecs} for a variety of intuitive arguments). In our case, there is no \textit{a priori} relation between $\tilde{\vect{v}}^L_{Lm}$ and $\vect{h}$ that would compel them to lie in some joint smaller-dimensional subspace, and since the Hamiltonian is $K$-local, the effective dimension $D$ of $\vect{h}$ is $\sum_{k=1}^K {n\choose k}=O(n^K)$ as discussed in \cref{ss:bound-H}. 
This polynomially small overlap is, however, determined by the circuit width $n$ rather than its depth $L$, so it will not be an NIBP as per \cref{def:NIBP}. See \cref{app:BD-proof} for a proof sketch and simulation results. Of course, any circuit for which width and depth are related will impose a barren plateau-type result via \cref{eq:non-unital-cost} that does depend on $L$; this type of noise-free barren plateau is the one first demonstrated in Ref.~\cite{McClean2018}.

\subsubsection{No guarantee of an NIBP: general argument}

We now show that the example above can be generalized, in particular without assuming that $l=L$. Specifically:

\begin{mytheorem}
\label{th:non-unital-no-NIBP}
Assume that a circuit described by \cref{eq:fullcircuit} contains any sequence of HS-contractive non-unital maps $\{\mc{N}_i\}_{i=1}^{l-1}$ for which $\max_{1\le i \le l-1} \sigma_{\max}(M_i) \in [0,\mu)$, where $\mu\in(0,1/2]$, $l\ge 3$, and $\sigma_{\max}(M_i)$ is the largest singular value of $M_i$, the real rotation+dilation matrix associated with the map $\mcN_i$. 
Assume in addition that the last maps $\{\mcN_i\}_{i=l}^L$ in the circuit are all HS-contractive non-unital, $L-l=O(1)$ for large $L$, and $\{\sigma_{\min}(M_i)>c\}_{i=l}^L$, where $c>0$ and $\sigma_{\min}(M_i)$ denotes the smallest singular value of $M_i$.
The cost function of such a circuit does not exhibit an NIBP.
\end{mytheorem}

Before we present the proof, we remark that \cref{th:concentration} concerns the behavior of the \emph{cost function} in circuits deeper than logarithmic depth ($L\in\omega[\log(n)]$), while \cref{th:non-unital-no-NIBP} concerns the behavior of \emph{derivative of the cost function} in the last few layers of deep circuits (large $L$). The class of deep circuits in \cref{th:non-unital-no-NIBP} is included in the class considered in \cref{th:concentration} since deep circuits satisfy $L\in\omega[\log(n)]$. The concentration of the cost function in \cref{th:concentration} does not prevent the non-concentration of the derivative in \cref{th:non-unital-no-NIBP}. This is because the cost function in \cref{th:concentration} concentrates within a \emph{finite interval} (the NILS;  see \cref{sec:NILS-subsec}), which means that a pair of cost functions $C(\vect{\theta}_1)$ and $C(\vect{\theta}_2)$ for circuits of the same depth $L$ can have a finite difference. Using PSRs, this allows for scenarios where the derivative of the cost function in \cref{th:non-unital-no-NIBP} remains finite, i.e., non-concentrated.

\begin{proof}
We start from \cref{eq:v_j<l}.
There is no additional bifurcation after the $l$’th layer, so that using \cref{eq:v_l^pm}, and the recursion given by \cref{eq:v_l} again, we obtain for $1 \le j \le L-l$:
\bes
\label{eq:v_j>l}
\begin{align}
\label{eq:v_j>l-1}
\vect{v}_{lm}^{\pm,l+j}  &= \O_{l+j} \cdots \O_{l+1}\vect{v}^{\pm,l}_{lm} + \vect{d}_{l+j} \\
\label{eq:v_j>l-2}
\vect{d}_{l+j} &=  \O_{l+j} \cdots \O_{l+2} \vect{c}_{l+1}+  \cdots \notag \\
&\qquad\qquad+ \O_{l+j} \vect{c}_{l+j-1}+\vect{c}_{l+j} .
\end{align}
\ees
Note that it follows from \cref{eq:lemma3-b} applied to \cref{eq:v_j<l-1,eq:v_j>l-1} that
$\|\vect{d}_j\| < 1$ for all $1\le j \le L$.

Combining \cref{eq:v_j<l,eq:v_j>l} we obtain:
\bes
\label{eq:vwef}
    \begin{align}
\label{eq:vwef-1}
        &\vect{v}_{lm}^{\pm,L} \notag\\
        &= \O_{L} \cdots \O_{l+1}\vect{v}^{\pm,l}_{lm} + \vect{d}_{L}\\
 \label{eq:vwef-2}
       &= \O_{L} \cdots \O_{l+1}(\O_{lm}^{\pm} \vect{v}_{l-1} +\vect{c}_l) + \vect{d}_{L}\\
        &= \O_{L} \cdots \O_{l+1}[\O_{lm}^{\pm}(\O_{l-1} \cdots\O_1\vect{v}_{0} +\vect{d}_{l-1}) +\vect{c}_l] \notag\\
        &\qquad\qquad\qquad\qquad\qquad\qquad\qquad\qquad\qquad+ \vect{d}_{L}\\
        \label{eq:vwef-3}
        &= \underbrace{\O_{L} \cdots \O_{l+1}\O_{lm}^{\pm}\O_{l-1}\cdots\O_1 \vect{v}_{0}}_{\vect{w}_{lm}^{\pm,L}}\\ \notag
        &\quad+ \underbrace{\O_{L} \cdots \O_{l+1}\O_{lm}^{\pm}\vect{d}_{l-1}}_{\vect{e}_{lm}^{\pm,L}}  +\underbrace{\O_{L} \cdots \O_{l+1}\vect{c}_l+\vect{d}_{L}}_{\vect{f}_L} .
    \end{align}
\ees
Let $\tilde{\vect{w}}_{lm}^{L} \equiv \vect{w}_{lm}^{+,L} - \vect{w}_{lm}^{-,L}$ and $\tilde{\vect{e}}^L_{lm} \equiv \vect{e}_{lm}^{+,L} - \vect{e}_{lm}^{-,L}$. Then, using \cref{eq:vwef}: 
\beq
\label{eq:vlmtilde}
\tilde{\vect{v}}^L_{lm} \equiv \vect{v}_{lm}^{+,L} - \vect{v}_{lm}^{-,L} =  \tilde{\vect{w}}^L_{lm} +\tilde{\vect{e}}^L_{lm} .
\eeq 

We will show that $\tilde{\vect{w}}^L_{lm}$ can be neglected but $\tilde{\vect{e}}^L_{lm}$ cannot. First:
\bes
\begin{align}
&\|\tilde{\vect{w}}_{lm}^{L}\| \notag\\
&= \| \O_L\cdots \O_{l+1}(\O_{lm}^+ - \O_{lm}^-)\O_{l-1}\cdots \O_1\vect{v}_0 \| \\
&\le \| \O_L\cdots \O_{l+1}(\O_{lm}^+ - \O_{lm}^-)\O_{l-1}\cdots \O_1\|\|\vect{v}_0\| \\
&< \| \O_L\| \cdots \|\O_{l+1}\| \|\O_{lm}^+ - \O_{lm}^-\| \| \O_{l-1}\| \cdots \|\O_1\| \\
&< 2 p_L \cdots p_1
\end{align}
\ees
where the second line follows by definition of the operator norm, in the third line we used submultiplicativity and $\|\vect{v}_0\|<1$, and in the last line we defined $p_l \equiv \| M_l\|$ and used $\|\O_{lm}^+ - \O_{lm}^-\| = \|M_l (O_{lm}^+ - O_{lm}^-)\|$ and $\|(O_{lm}^+ - O_{lm}^-)\|\le 2$ since $O_{lm}^\pm$ are orthogonal and the eigenvalues of any orthogonal matrix are all $\pm 1$. We have $p_l < 1$ by \cref{lemma2} and \cref{lemma3}.
We may thus write
\beq
\label{eq:wtilde-bound}
\| \tilde{\vect{w}}_{lm}^{L}\| < 2 p^L \ , \quad p<1 ,
\eeq
where $p \equiv (\prod_{l=1}^L p_l)^{1/L}$. Thus, $\| \tilde{\vect{w}}_{lm}^{L}\|$ vanishes exponentially in the circuit depth.

Next, let us consider $\tilde{\vect{e}}^L_{lm}$. We can rewrite \cref{eq:v_j<l-2} for $j=l-1$ as:
\beq
\vect{d}_{l-1} = \vect{c}_{l-1}+\sum_{l'=1}^{l-2} \prod_{i=l-1}^{l'+1} \O_i \vect{c}_{l'}\ , \quad l\ge 3 \ ,
\label{eq:d_l-1}
\eeq
where the order in the product reflects operator ordering.

Next, we lower-bound $\|\vect{d}_{l-1}\|$. As a simple example for which $\|\vect{d}_{l-1}\|$ is lower bounded by a positive constant, consider the case where only the last map is non-unital, and the rest are unital, i.e., $\vect{c}_{l-1} \ne \vect{0}$ but $\{\vect{c}_{j} = \vect{0}\}_{j=1}^{l-2}$. Then $\|\vect{d}_{l-1}\| = \|\vect{c}_{l-1}\| >0$.
To make the argument more general, we note that it follows from the triangle inequality $\|A+B\| \ge |\|A\|-\|B\||$ that:
\bes
\label{eq:triangle}
\begin{align}
\label{eq:triangle-a}
\|\sum_{k=1} A_k \| &\ge | \|A_1\| - \|\sum_{k=2} A_k \| | \\ 
\label{eq:triangle-b}
& \ge \|A_1\| - \|\sum_{k=2} A_k \|  \ge \|A_1\| - b ,
\end{align}
\ees
where $b \ge \|\sum_{k=2} A_k \|$.
In the context of bounding $\|\vect{d}_{l-1}\|$, we identify 
\beq
A_1\equiv \vect{c}_{l-1}\ , \quad \{A_k\}_{k=2}\equiv\left\{\prod_{i=l-1}^{l'+1} \O_i \vect{c}_{l'}\right\}_{l'=1}^{l-2} .
\label{eq:A_i}
\eeq 
Using the polar decomposition $M_i = V_i S_i$ with $V_i$ orthogonal and $S_i = \sqrt{M_i^\dag M_i}$ positive semidefinite and symmetric, we have 
\begin{align}
&\| \O_i \vect{c}\| = \|V_i S_i O_i\vect{c}\| = \|S_i O_i\vect{c}\|\\
&\qquad\qquad\qquad\qquad\qquad\leq \lambda_{i,1} \|O_i\vect{c}\| 
= \lambda_{i,1} \|\vect{c}\| ,\notag
\end{align}
where, using \cref{lemma3}, $\lambda_{i,1} = \sigma_{\max}(M_i) <  1$ is the largest singular value of $M_i$, and $\|\vect{c}\| < 1$. The same calculation, pulling out one factor from the left at a time, yields:
\beq
 \| \prod_i  \O_i \vect{c}\| \le \prod_i \lambda_{i,1}\|\vect{c}\| .
\eeq
Thus, letting 
\beq
\tilde{\lambda}_l \equiv \max_{1\le i \le l-1} \lambda_{i,1}\ , \quad \tilde{c}_l\equiv \max_{1\le l'\le l-1}\|\vect{c}_{l'}\|\ ,
\label{eq:maxima}
\eeq
we have:
\bes
\label{eq:63}
\begin{align}
\label{eq:63-1}
& \| \sum_{l'=1}^{l-2} \prod_{i=l-1}^{l'+1} \O_i \vect{c}_{l'}\|  \le \sum_{l'=1}^{l-2} \| \prod_{i=l-1}^{l'+1}  \O_i \vect{c}_{l'}\| \\
\label{eq:63-2}
&\qquad \le \sum_{l'=1}^{l-2} \prod^{l-1}_{i=l'+1} \lambda_{i,1} \|\vect{c}_{l'}\| 
\le \sum_{l'=1}^{l-2} \|\vect{c}_{l'}\|  \tilde{\lambda}_l^{l-l'-1} \\ 
\label{eq:63-3}
& \qquad \le  \tilde{c}_{l} r_l \ , \qquad r_l \equiv \frac{\tilde{\lambda}_l-\tilde{\lambda}_l^{l-1}}{1-\tilde{\lambda}_l}\ , \quad l \ge 3 ,
\end{align}
\ees
where in the first inequality in \cref{eq:63-2} we inverted the order in the product since $l'+1 \le l-1$. Using \cref{eq:d_l-1,eq:triangle,eq:A_i} we thus have:
\bes
\begin{align}
    \|\vect{d}_{l-1}\| &\ge | \|\vect{c}_{l-1}\| - \| \sum_{l'=1}^{l-2} \prod_{i=l-1}^{l'+1} \O_i \vect{c}_{l'}\| | \\
    & \ge \|\vect{c}_{l-1}\| - \tilde{c}_{l} r_l .\label{eq:norm-d_l-1-bound}
\end{align}
\ees

We defined $\tilde{c}_l$ so that $\|\vect{c}_{l-1}\|$ is included in the maximum in \cref{eq:maxima}. Therefore, if $\tilde{c}_l=\|\vect{c}_{l-1}\|$ and $r_l < 1$, then the r.h.s. of \cref{eq:norm-d_l-1-bound} is positive. The condition $r_l < 1$ holds for all $\tilde{\lambda}\in [0,1/2)$. If, instead, $\tilde{c}_l$ corresponds to $\max_{1\le l'\le l-1}\|\vect{c}_{l'}\|$ with $l'<l-1$, then we still have $\|\vect{d}_{l-1}\| > 0$, provided $r_l < \|\vect{c}_{l-1}\|/\tilde{c}_l$; this condition is satisfied for all $\tilde{\lambda}_l\in [0,\mu)$, where $\mu\in(0,1/2)$ is found by solving the transcendental equation $r_l = \|\vect{c}_{l-1}\|/\tilde{c}_l$ for $\tilde{\lambda}_l$.
Thus, we may conclude that a sufficient condition for $\|\vect{d}_{l-1}\| > C$, where $C>0$ is a constant, is that the circuit contains any sequence of HS-contractive non-unital maps $\{\mc{N}_i\}_{i=1}^{l-1}$ for which $\max_{1\le i \le l-1} \sigma_{\max}(M_i) \in [0,\mu)$ and $l\ge 3$, where $\sigma_{\max}(M_i)$ is the largest singular value of $M_i$, i.e., the largest eigenvalue of the dilation $S_i = \sqrt{M_i^\dag M_i}$ corresponding to $\mc{N}_i$.

Next, from \cref{eq:vwef-3} we need to consider $\O_{lm}^\pm \vect{d}_{l-1} = M_l O_{lm}^\pm \vect{d}_{l-1}$. The two vectors $O_{lm}^\pm \vect{d}_{l-1}$ are separated by a distance $d_l = \|(O_{lm}^+ - O_{lm}^-) \vect{d}_{l-1}\|$ equal to the sum of the norms of the orthogonal vectors to their projections onto $\vect{d}_{l-1}$, i.e., 
\bes
\label{eq:norm-d_l-1-bound-2}
\begin{align}
d_l &= 2\|\vect{d}_{l-1}\||\sin(\theta/2)| \\
\cos\theta &= \frac{(O_{lm}^+ \vect{d}_{l-1})\cdot (O_{lm}^- \vect{d}_{l-1})}{\|\vect{d}_{l-1}\|^2}\ . 
\end{align}
\ees
The remaining transformations prescribed by \cref{eq:vwef-3} are $\O_L\cdots\O_{l+1}M_l$, where $\O_i = M_i O_i$ and, using the polar decomposition again, $M_i = V_i S_i$ with $V_i$ orthogonal and $S_i = \sqrt{M_i^\dag M_i}$. The orthogonal rotations preserve the norms of $O_{lm}^\pm \vect{d}_{l-1}$; the dilations $\{S_i\}_{i=l}^L$ shrink these norms by at most the products of their smallest singular values, $\sigma_{\min}(M_i)$. Thus $d_l \mapsto \sigma_{\min}(M_l)\cdots\sigma_{\min}(M_L)d_l$, and as a result:
\beq
\|\tilde{\vect{e}}^L_{lm}\| \ge  \sigma_{\min}(M_l)\cdots\sigma_{\min}(M_L)d_l \ .
\eeq
The final step is to use \cref{eq:wtilde-bound,eq:vlmtilde} and the triangle inequality to write
\bes
\label{eq:final-bound}
\begin{align}
\label{eq:final-bound-1}
\| \tilde{\vect{v}}^L_{lm} \| &= \| \tilde{\vect{w}}^L_{lm} +\tilde{\vect{e}}^L_{lm}\| \ge  \| \tilde{\vect{e}}^L_{lm}\| - \| \tilde{\vect{w}}^L_{lm}\|  \\
\label{eq:final-bound-2}
&\ge \sigma_{\min}(M_l)\cdots\sigma_{\min}(M_L)d_l - 2p^L   .
\end{align}
\ees
Therefore, as long as 
\beq
\{\sigma_{\min}(M_i)>c\}_{i=l}^L \text{ and } L-l\in O(1) ,
\eeq 
then $\| \tilde{\vect{v}}^L_{lm}\|>C$ where $c,C>0$ are both constants. This is because $2p^L$ is exponentially small in $L$, and $\sigma_{\min}(M_l)\cdots\sigma_{\min}(M_L)d_l$ is a product of $O(1)$ positive constants, i.e., itself a positive constant. Thus, $\| \tilde{\vect{v}}^L_{lm}\|$ is lower bounded by a positive constant for $L$ sufficiently large: $L > \log(a^\k d_l/2)/\log(p)$, where $\k=L-l>0$ is an $O(1)$ constant and $a = (\Pi_{i=l}^L\sigma_{\min}(M_i))^{\frac{1}{L-l}}>0$ is another constant.

Reverting to \cref{eq:gradient-result}, the fact that the lower bound on $\|\tilde{\vect{v}}^L_{lm} \|$ is now a positive constant for any $l$ such that $L-l=O(1)$ then means that there is no NIBP in the HS-contractive non-unital case. As in the case of \cref{eq:non-unital-cost}, the overlap in \cref{eq:gradient-result} could still be exponentially small (Levy's lemma), but as argued above, this is not noise-induced.
\end{proof}

\cref{th:non-unital-no-NIBP} assumes that $\sigma_{\min}(M_i)>0$. However, there exist non-unital noise channels for which $\sigma_{\min}(M_i)=0$, e.g., a composite of a bit-flip or phase-flip channel with Kraus operators $\{\sqrt{p}I,\sqrt{1-p}\sigma\}$ where $\sigma\in\{X,Y,Z\}$, at the special symmetry point $p=1/2$, followed by an amplitude damping channel in \cref{eq:AD}; see \cref{app:sigma_min}.  We next discuss the possibility of NIBPs in such cases.

When $\sigma_{\min}(M_i) = 0$, there is an in-principle possibility of encountering NIBPs in circuits subject to non-unital noise. To see why the proof of \cref{th:non-unital-no-NIBP} does not hold in this case, note that $\| \tilde{\vect{v}}^L_{lm} \|$ in \cref{eq:final-bound-2} is not lower bounded by $C>0$ if at least one of $\{\sigma_{\min}(M_i)>c\}_{i=l}^L$ is zero. Hence, there is no guarantee that the circuit will escape an NIBP since $\| \tilde{\vect{v}}^L_{lm} \|$ is not lower bounded above zero.

However, the contraction of $\vect{d}_l$ by $\sigma_{\min}(M_i)$ only happens when $\vect{d}_l$ is aligned in specific directions such that the noise channel maps $\vect{d}_l$ to the zero vector. As the number of qubits in the circuit grows, the dimension of $\vect{d}_l$ grows exponentially, making this scenario highly unlikely by a similar argument to the one below \cref{eq:non-unital-cost}, using the standard Levy's lemma. Hence, we expect that NIBPs are avoided under non-unital noise even for channels for which one or more $\sigma_{\min}(M_i) = 0$.

\subsubsection{Example: amplitude-damping}
\label{sec:AD}

As a physical example of a HS-contractive non-unital map that prevents an NIBP, consider an amplitude-damping map. It suffices to consider the simple case of the amplitude-damping map for a qubit coupled to a zero-temperature bath. The Kraus operators are 
\beq
\label{eq:AD}
K_0 = \ketb{0}{0} + \sqrt{1-p}\ketb{1}{1} , \quad K_1 = \sqrt{p} \ketb{0}{1} , 
\eeq
where $p$ is the probability of relaxation from the excited state $\ket{1}$ to the ground state $\ket{0}$~\cite{Lidar:2019aa}. We find, using \cref{eq:M-c}, that $\vect{c} = (0,0,p)$ and $M = \text{diag}(\sqrt{1-p},\sqrt{1-p},1-p)$, so that $\sigma_{\max}(M) = \sqrt{1-p}$ and $\sigma_{\min}(M) = 1-p$. Thus, according to \cref{th:non-unital-no-NIBP}, for any $p \in [3/4+\eps,1-\eps]$ and $\eps>0$, a sequence of $L$ such amplitude-damping maps acting independently on each qubit of a noisy VQA circuit will prevent the cost function of such a circuit from exhibiting an NIBP.

\section{Simulations}
\label{sec:sim}

We performed numerical simulations employing the Qiskit framework~\cite{Qiskit} to ascertain the ground state energy of specific Hamiltonian under the influence of depolarizing (unital) and amplitude-damping (HS-contractive non-unital) noise. The single-qubit depolarizing map is defined as: 
\beq
\mathcal{N}(\rho) = 
(1-p)\r+ \frac{p}{3} \sum_{\alpha \in \{x,y,z\}} \sigma^\alpha \r \sigma^\alpha ,
\eeq 
where $p$ is the probability of the error. The single-qubit Kraus operators of the amplitude-damping map are given in \cref{eq:AD}.

The simulation was conducted utilizing three-qubit VQAs incorporating the \texttt{TwoLocal} ansatz, composed of $RY(\theta)=\exp(-i\theta \s^y/2)$ rotation gates and CNOTs in a linear entanglement configuration, where qubit $i$ is entangled with qubit $i+1$ $\forall i$ along a chain, with open boundary conditions. The optimization procedure was performed using the Stochastic Perturbation Simulated Annealing (SPSA) algorithm~\cite{spsa-pdf} as the classical optimizer, constrained by a maximum iteration limit of 200 (\texttt{maxiter=200}). The multi-qubit noise map employed in the simulation was constructed from a composition of $n$ one-qubit noise maps, which are HS-contractive for both of the examples we used to represent unital and non-unital maps.

We employ a stochastic procedure to generate a set of $50$ sparse three-qubit Hamiltonians. These Hamiltonians are structured as $H = \sum_{i_1,i_2,i_3} h_{i_1 i_2 i_3} \sigma_{i_1} \otimes \sigma_{i_2} \otimes \sigma_{i_3}$, with the constraint that each element $\sigma_{i_j}$ is drawn from the set ${\sigma^0, \sigma^x, \sigma^z}$ and interactions are restricted to be two-local, i.e., there is at least one $\sigma_{i_j}$ that is $\sigma^0$. The magnitude of $h_{i_1 i_2 i_3}$ is uniformly sampled from $[0,1)$ and later normalized such that $\|H\|_2=1$. The Hamiltonian for the main simulation was restricted to a maximum of three qubits to reduce the effect of non-noise-induced forms of BP. 

In the simulation to demonstrate other types of BP that are not NIBP, we performed similar simulations using $50$ randomly generated $n$-body Hamiltonians $H_n$, with $2\le n \le 9$. We constructed these Hamiltonians so that each has a zero ground state energy and is at most two-local.
These Hamiltonians are represented in the format $H_n = \sum_{\vect{i}} h_{\vect{i}} \sigma_{i_1} \otimes \cdots \otimes \sigma_{i_n}$, where $\sigma_{i_j}$ belong to the set ${\sigma_0, \sigma_x, \sigma_z}$ and $\vect{i}=(i_1,\dots,i_n)$. 

The pseudo-code to generate random $n$-qubit Hamiltonians in this simulation is given in \cref{algo:randH}.

\begin{algorithm}
  \caption{Generating Hamiltonians for simulations}\label{algo:ham}
  \begin{algorithmic}[1]
    \Procedure{Generating a 2-local Hamiltonian}{}
      \State Requirement: $H_0=0$
      \State Randomly generate a 2-local Hamiltonian $H$ of $n$ qubits such that $\|H\|_2=1$.
      \State Find its ground state $H_0$.
      \State Rescale the energy of $H$ to have $H_0=0$ by subtracting the ground state energy.
      \State Renormalize $H$ such that $\|H\|_2=1$.
    \EndProcedure
  \end{algorithmic}
  \label{algo:randH}
\end{algorithm}

\subsection{Demonstration of Theorem \ref{th:unital-NIBP} and \ref{th:non-unital-no-NIBP}}

NIBP refers to the phenomenon where the magnitude of the cost function's gradient with respect to control angles diminishes exponentially in the number of layers [\cref{eq:NIBP1}]. Our analysis demonstrates this characteristic to be consistently applicable solely within the unital noise scenario (\cref{th:unital-NIBP}). Conversely, in the case of HS-contractive non-unital noise, the magnitude of the gradient need not necessarily experience exponential suppression as a function of the number of layers (\cref{th:non-unital-no-NIBP}).
We now present simulation results that support these results.

\begin{figure}[t]
        \includegraphics[width=0.24\textwidth]{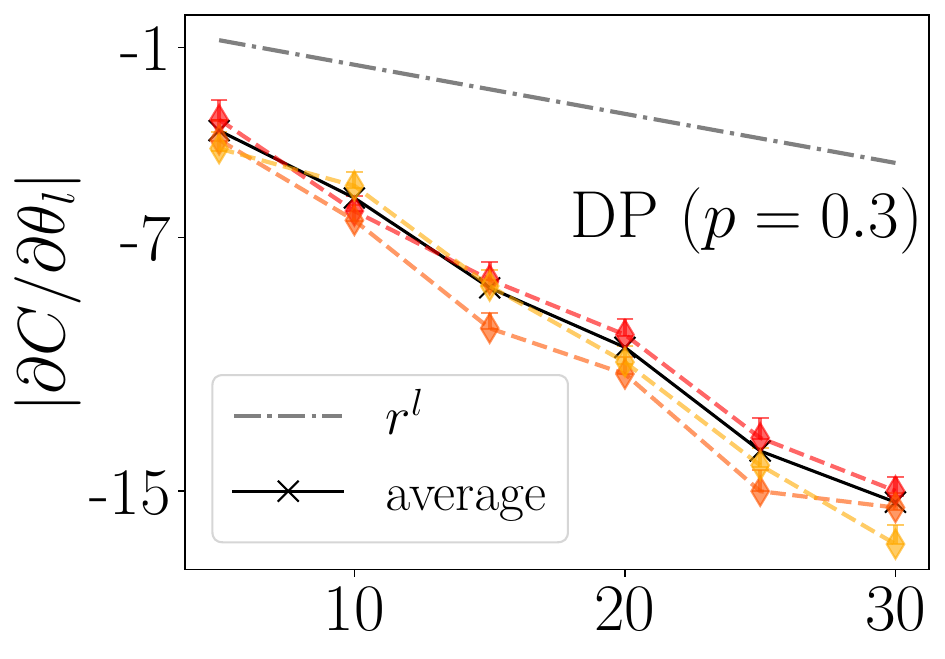}  
        \includegraphics[width=0.225\textwidth]{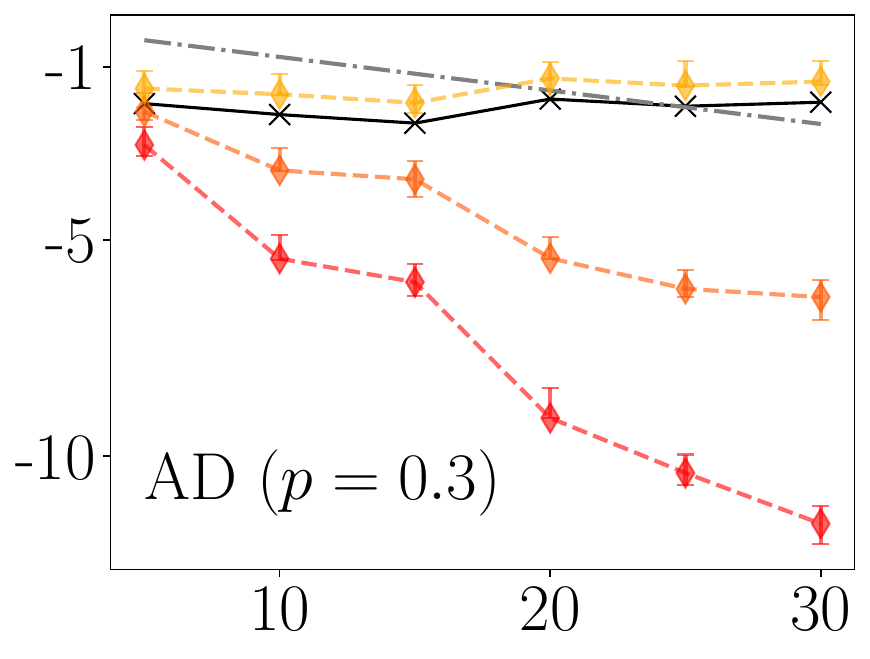}
        \includegraphics[width=0.24\textwidth]{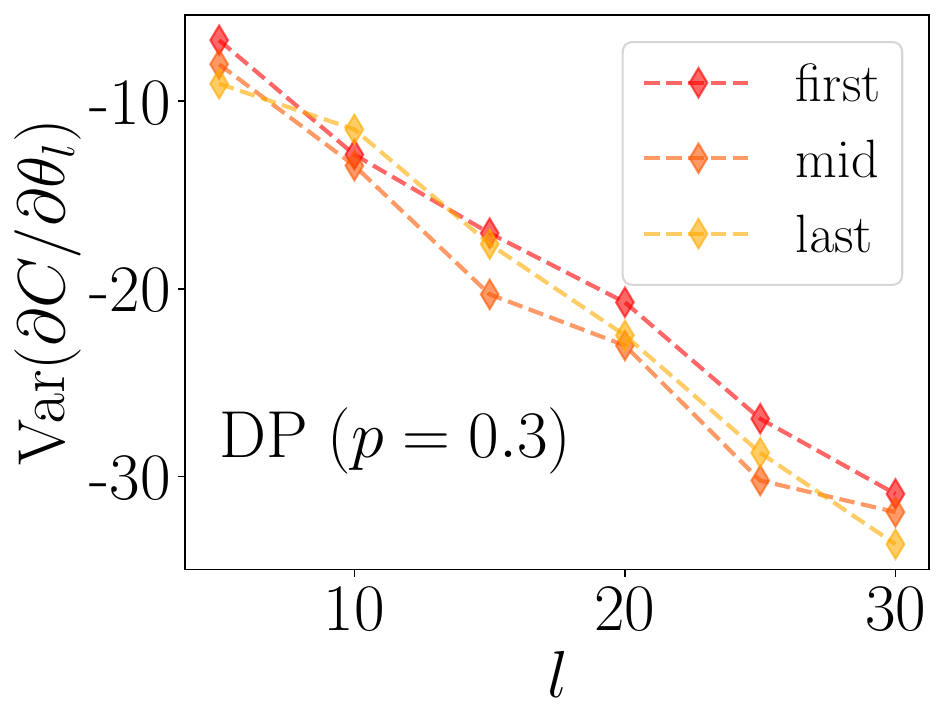} 
        \includegraphics[width=0.225\textwidth]{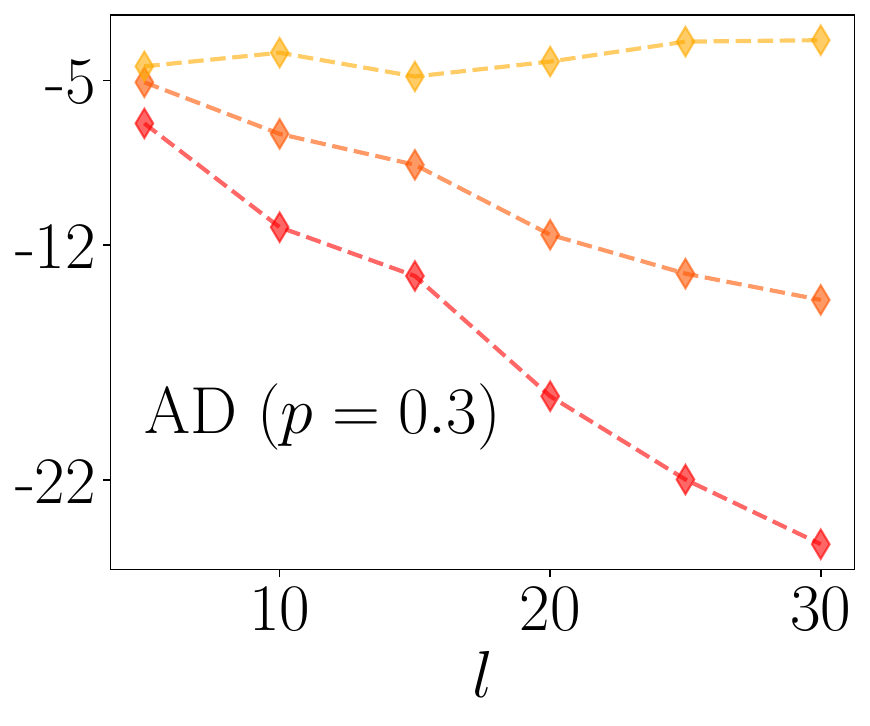}  
    \caption{Mean and variance ($\log_{10}$ scale) of the magnitude of the cost function gradient for depolarizing (left) and amplitude-damping (right) maps with noise probability $p=0.3$ as a function of layers. Error bars in the upper plots represent the range between the maximum and minimum values.} 
    \label{fig:layer-DP-AD}
\end{figure}

\cref{fig:layer-DP-AD} illustrates the magnitude of the cost function gradient (on a logarithmic scale) as a function of layers of the VQAs. Mean values (variances) are plotted in the top (bottom) two plots. Plots on the left (right) correspond to VQAs under depolarizing noise (amplitude-damping noise), both with a noise probability of $p=0.3$. We specifically present three angles corresponding to the initial, middle, and final layers of the VQA to emphasize their distinct behaviors. Additionally, the dot-dashed black line, denoted as $r^l$, is in proportion to the bound derived in \cref{eq:NIBP1} for the unital case. The solid black line is the numerical average over the three angles shown.

Under unital, depolarizing noise (upper and lower left), both the mean and variance of the gradient magnitude exhibit an exponential decay with an increasing number of layers. The error bars displayed in the upper plots represent the range between the maximum and minimum values. This suggests that all simulated mean values remain within the predicted bound. The observed behavior remains consistent irrespective of the specific layer within the VQA from which these angles were selected. The theoretical upper bound (dot-dashed line) is rather loose but consistent with the numerical results.

Under amplitude-damping noise (upper and lower right), distinct behaviors are observed among the three angles. The mean value of the magnitude of the gradient increases as the angle approaches the end of the VQA. The angle selected from the final layer consistently demonstrates a nearly constant gradient magnitude and violates the bound in \cref{eq:NIBP1}, consistent with \cref{th:non-unital-no-NIBP} and as anticipated in our theoretical analysis that angles within the terminal layers of the VQA circuit under HS-contractive non-unital noise may evade the NIBP.

\begin{figure}
        \includegraphics[width=0.24\textwidth]{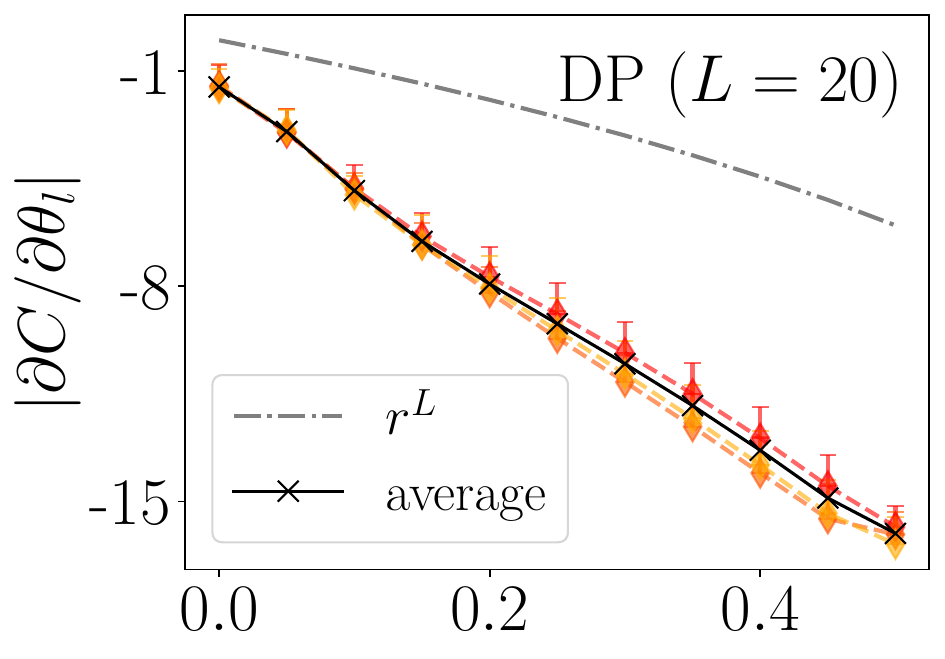}  
        \includegraphics[width=0.225\textwidth]{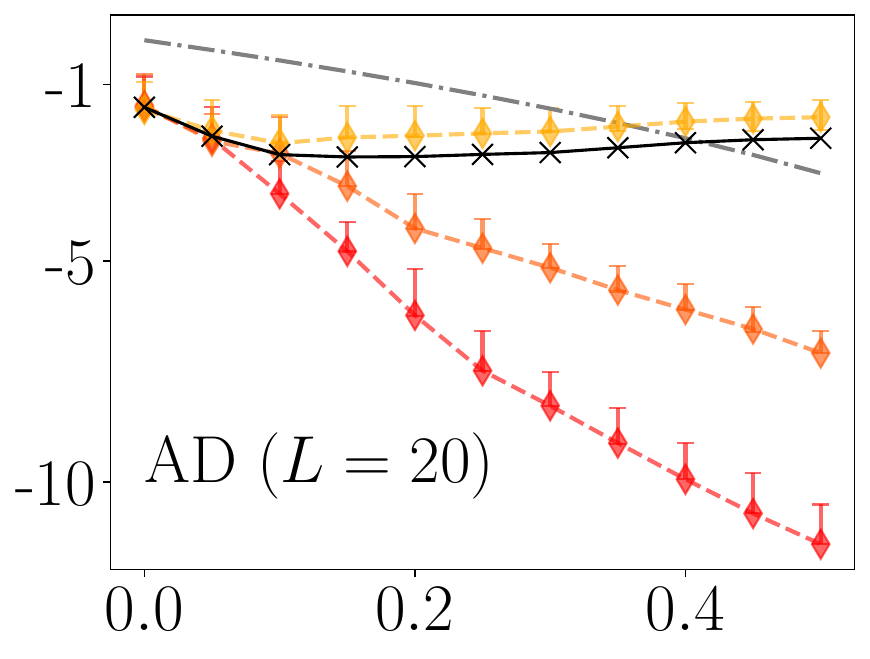}
        \includegraphics[width=0.24\textwidth]{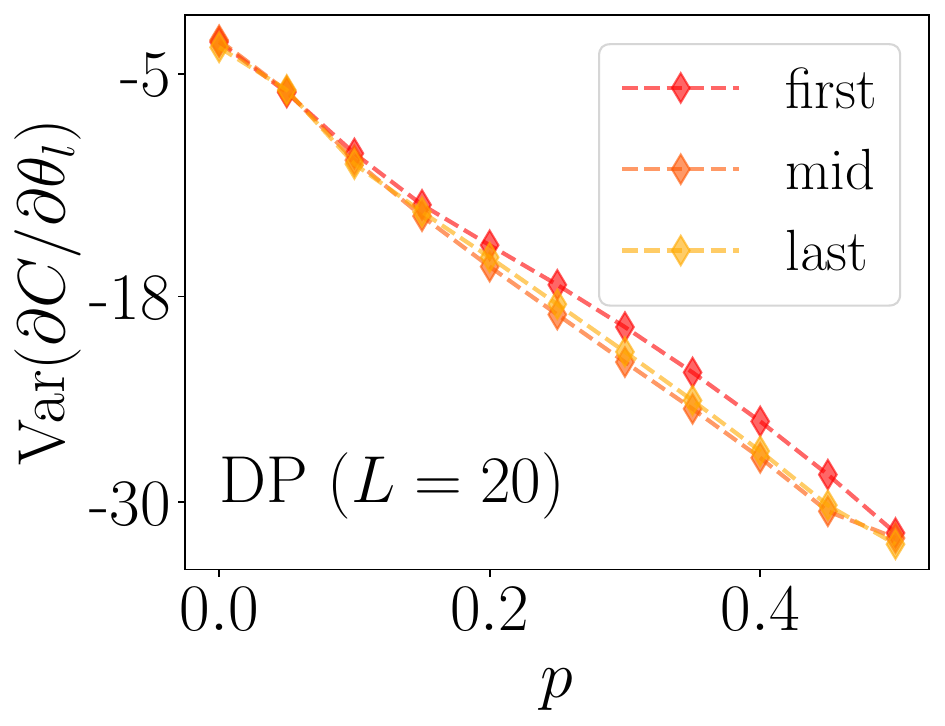} 
        \includegraphics[width=0.225\textwidth]{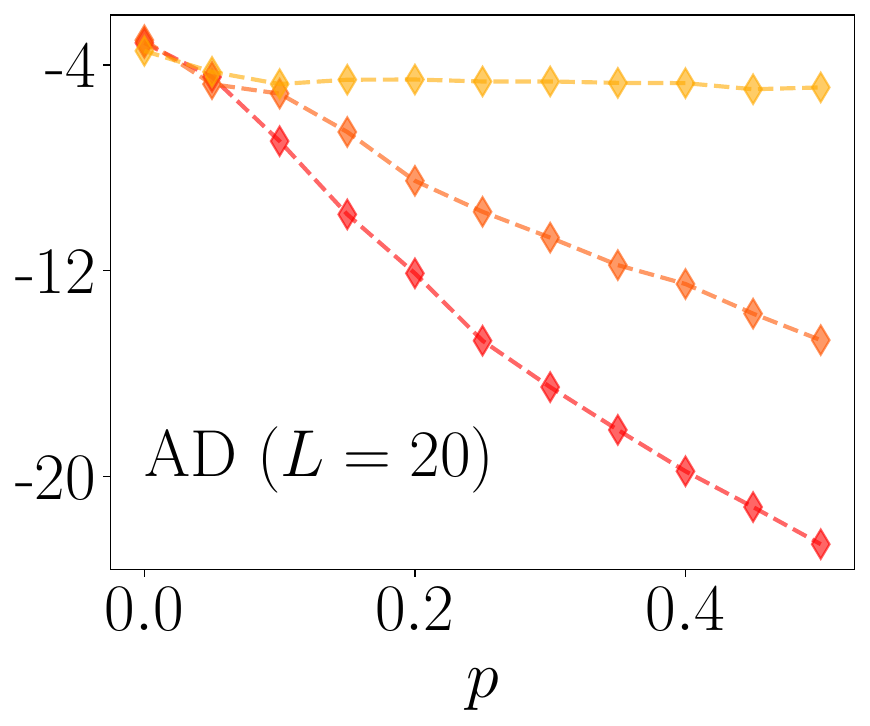}  
    \caption{Mean and variance ($\log_{10}$ scale) of the magnitude of the cost function gradient for depolarizing (left) and amplitude-damping (right) maps as a function of noise probability $p$ in a VQA with 20 layers. Error bars in the upper plots represent the range between the maximum and minimum values.} 
    \label{fig:bound-DP-AD}
\end{figure}

\cref{fig:bound-DP-AD} depicts an alternative view of \cref{eq:NIBP1}, where the magnitude of the gradient is plotted as a function of the noise probability. The total number of layers in the VQA is fixed at $L=20$. The noise probability is varied from $p=0$ to $p=0.5$, corresponding to decreasing the parameter $r$ in \cref{eq:NIBP1}. The behavior seen in \cref{fig:bound-DP-AD} is similar to \cref{fig:layer-DP-AD}, where the magnitude of the gradient fully respects the derived bound only under depolarizing noise.

By fixing the noise probability and varying the number of layers (\cref{fig:layer-DP-AD}) or fixing the number of layers and varying the noise probability (\cref{fig:bound-DP-AD}), our simulations consistently demonstrate that \cref{eq:NIBP1} is well-respected in VQAs under unital noise and violated in VQAs under HS-contractive non-unital noise. This observation aligns with our theoretical predictions in \cref{th:unital-NIBP} and \cref{th:non-unital-no-NIBP}. For a numerical examination of the standard (noise-free) BP, see \cref{app:BD-proof}.

\subsection{Cost function behavior}
Our results indicate that unital noise is subject to both NIBP and other forms of BP, and that, conversely, HS-contractive non-unital noise may potentially evade NIBP despite experiencing a comparable degree of other BPs. The primary question within the framework of VQA pertains to its efficacy in determining the ground state of a specific Hamiltonian. Consequently, in practical applications where the presence of noise is inevitable, the central concern revolves around identifying the type of noise that is comparatively less detrimental. 

\begin{figure}
        \includegraphics[width=0.24\textwidth]{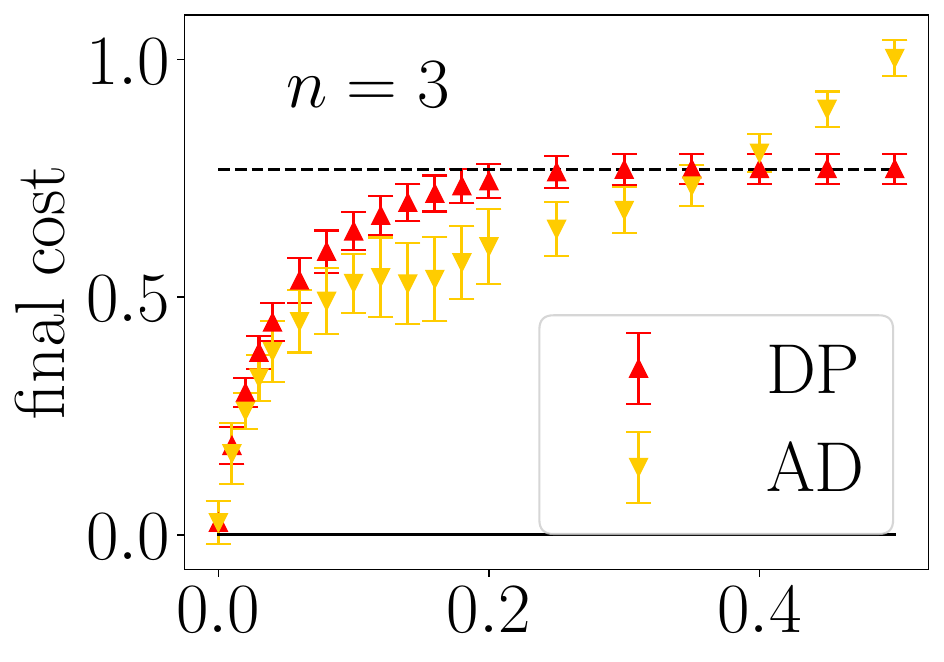}  
        \includegraphics[width=0.22\textwidth]{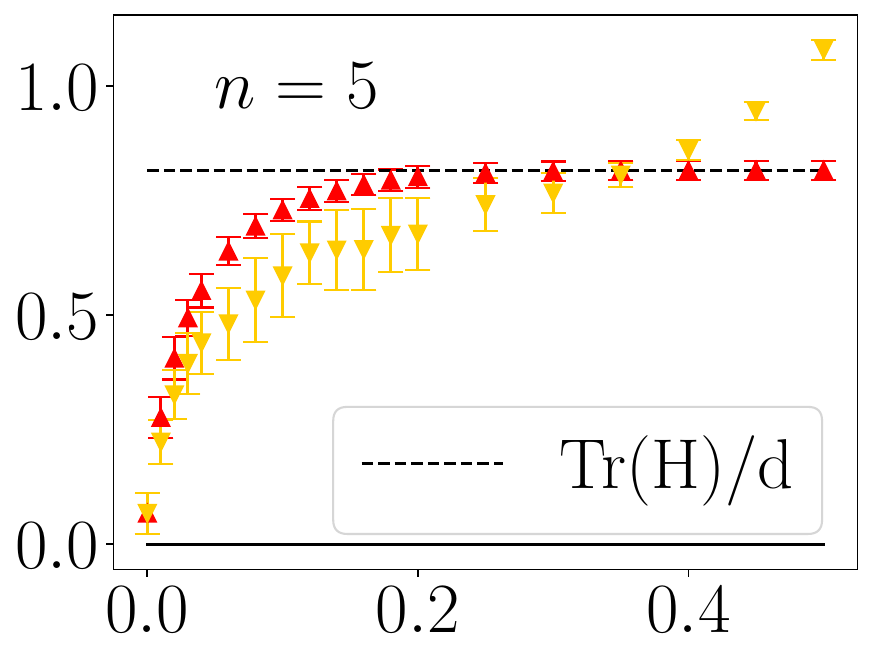}
        \includegraphics[width=0.24\textwidth]{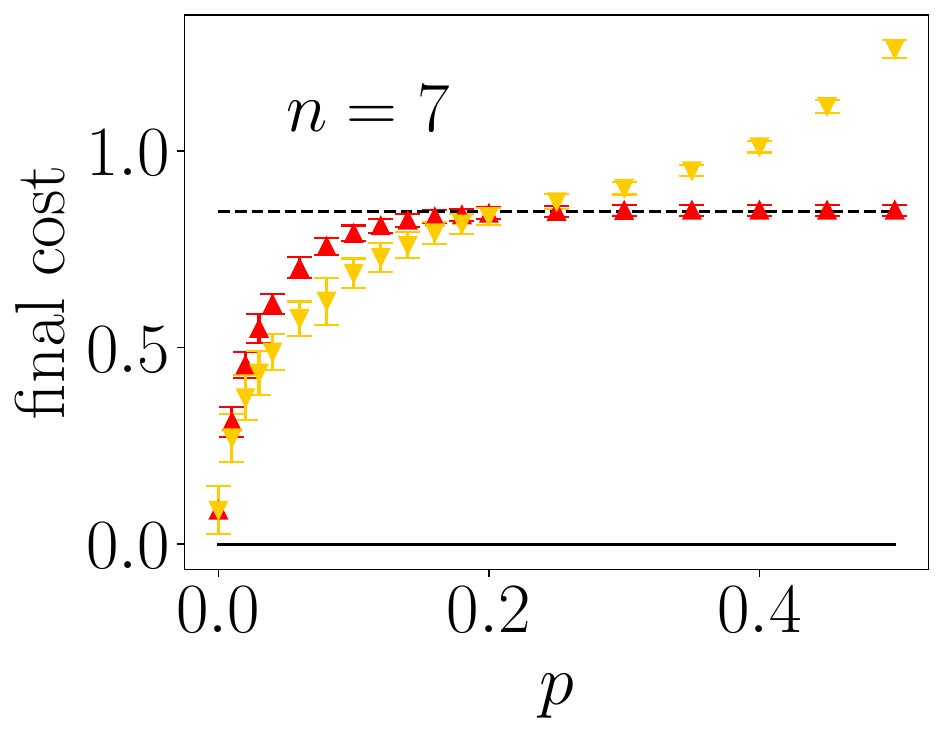} 
        \includegraphics[width=0.22\textwidth]{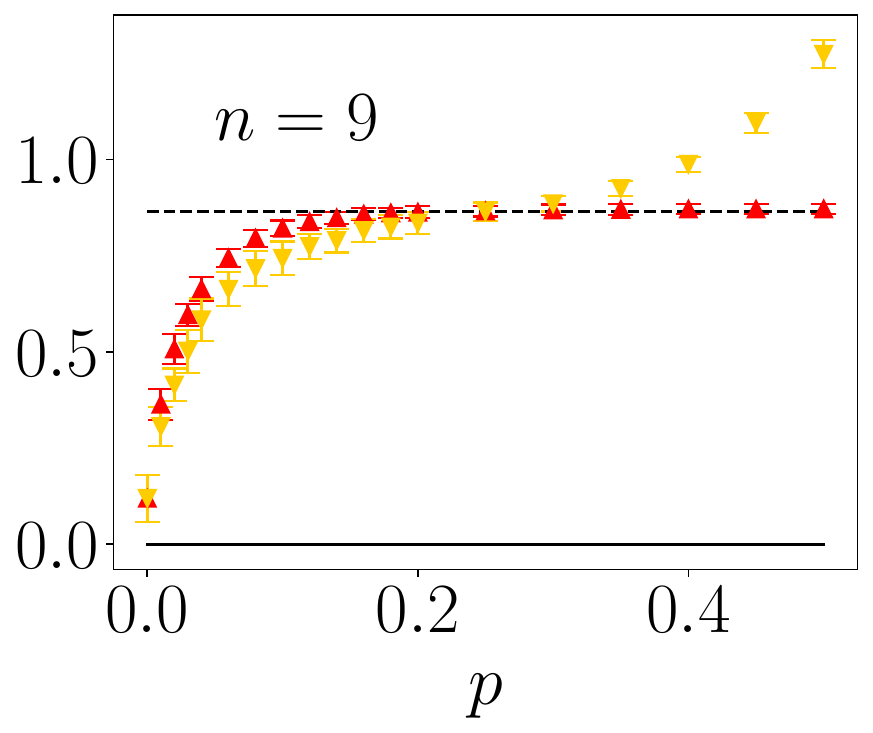}  
    \caption{Final cost function averaged over $50$ random $n$-qubit Hamiltonians with 
zero ground state energy under depolarizing (red up-triangles) and amplitude-damping (yellow down-triangles) maps as a function of noise probability $p$ using VQA with $L=5$ layers. The solid black line at zero denotes the true minimum of the Hamiltonians used in this simulation. The error bars are the standard deviation of the final cost. The dashed black line is the predicted NILS value in the large circuit depth limit in the unital case, \cref{eq:C_NILS-unital}.
} 
    \label{fig:final-cost-n}
\end{figure}

To address this question, we determined the final cost function obtained by VQA circuits when subjected to unital and HS-contractive non-unital noise maps. The result displayed in \cref{fig:final-cost-n} was achieved through the training of circuits responsible for generating the results illustrated in \cref{fig:layer-DP-AD,fig:bound-DP-AD}, showing the average of the instances for $n$-qubit Hamiltonians for $n=3, 5, 7$, and $9$ as a function of noise probability $p$.

The most obvious result from these simulations is that noise has a rapidly increasing, strongly detrimental effect: the final cost function rises rapidly from its optimal value of zero as soon as noise is introduced, whether unital or non-unital. Note that additionally, even at $p=0$, the final cost is greater than zero for $n\ge 5$ in our simulations, suggesting the effect of BPs. 

The difference between unital and non-unital noise is small within the low-noise regime. The non-unital case exhibits a slightly better performance for small $p$, until reaching a crossover point typically observed within the range of $p=0.15-0.3$ (depending on $n$). However, a notable divergence in their behavior becomes evident as the probability of noise increases. In the case of depolarizing noise, the cost function flattens out at approximately $p=0.2$, whereas the cost function continues to rise when subjected to amplitude-damping noise, accelerating around $p=0.3$. This observation seems consistent with our NILS result, where the cost function concentrates on a fixed point that is noise-independent in the unital case but noise-dependent in the HS-contractive non-unital case [\cref{eq:C_NILS}]. Indeed, the final cost function value precisely matches the theoretically predicted $\Tr(H)/d$ in the unital case [\cref{eq:C_NILS-unital}]. 

Note that the theoretical prediction is concerned with the large $L$ limit, not large $p$. However, while the simulations in \cref{fig:final-cost-n} are for fixed $L=5$, increasing $p$ at fixed $L$ is tantamount to increasing $L$ at fixed $p$, as is clear from \cref{eq:concentrate-C2}, where the factor $q^L$ is responsible for the NILS, and $q$ is the contractivity factor, which is, of course, monotonic in $p$.

Another measure of noisy circuit quality is trainability: the variance of the cost function gradient should vanish no faster than $\Omega(1/\text{poly}(n))$ \cite{Cerezo2021-bp}. \cref{eq:final-bound-2} implies that this is the case when $L-l=O(\log(n))$; see \cref{app:trainable} for details. It may appear that this implies an advantage for non-unital noise over unital noise. However, being trainable only implies the potential to be trained, and does not guarantee that the circuit can, in practice, be trained to achieve the global minimum. This is illustrated in \cref{fig:final-cost-n}, where the final cost does not converge to zero for $p>0$ due to the NILS.

\section{Conclusions}
\label{sec:conclusions}

This work expands the study of NIBPs to incorporate arbitrary unital and HS-contractive non-unital noise.  Using a generalization of the parameter shift rule that includes noise, we have derived upper bounds for the scaling of the magnitude of the cost function gradient with respect to circuit width $n$, circuit depth $L$, and noise strength. In the unital case, we have shown that the onset of an NIBP occurs already for circuits of logarithmic depth (\cref{th:unital-NIBP}). In contrast, in the HS-contractive non-unital case, we have shown that VQA circuits need not necessarily exhibit an NIBP. This is true, in particular, when a constant number of final layers in a VQA circuit are subject to HS-contractive non-unital noise (\cref{th:non-unital-no-NIBP}). 

We found that both unital and HS-contractive non-unital circuits exhibit a phenomenon we call a noise-induced limit set (NILS), whereby the cost function concentrates on a fixed value for circuits of greater than logarithmic depth. In the unital case, this is given by the expectation value of the problem Hamiltonian with respect to the fully mixed state [\cref{eq:C_NILS-unital}], but in the HS-contractive non-unital case, the fixed value is determined by the parameters of the noise map as well [\cref{eq:C_NILS}]. 

Our results are validated with numerical simulations for the depolarizing and amplitude-damping maps.
Interesting open questions we do not address here are whether HS-contractive ($\|M\|<1$) or non-HS-contractive $(\|M\|\ge 1)$ non-unital noise maps appear more often in practical scenarios, how to characterize them according to a given set of Kraus operators, and what the measure of HS-contractive non-unital maps is in the space of all non-unital maps.

Overall, our work shows that NIBPs present a significant challenge for VQAs, even after mitigation of the standard (noiseless) BP problem. A combination of error suppression, mitigation, and correction methods will be necessary to realize the promise of VQAs, just as is the case for other quantum algorithms running on noisy quantum computers.

\textit{Note added}. After this work was posted to the arXiv a related work appeared that addresses non-unital maps and barren plateaus~\cite{mele2024noiseinduced}.
The non-unital maps considered in this work are $n$-qubit tensor-products maps of $n$ one-qubit non-unital channels. Our results hold for any $n$-qubit non-unital HS-contractive maps ($\|M\|<1$), regardless of whether the map is a tensor-product of one-qubit channels.

\acknowledgments
This research was supported by the ARO MURI
grant W911NF-22-S-0007. This research was developed with funding from the Defense Advanced Research Projects Agency under Agreement HR00112230006 and Agreement HR001122C0063. We thank Rub\'{e}n Ibarrondo, Elias Zapusek and Samson Wang for useful comments, and especially Victor Kasatkin for numerous insightful discussions.

\bibliographystyle{quantum}

\onecolumn\newpage
\appendix

\section{Proof of \cref{eq:v-bounded}}
\label{app:purity-proof}

Using the purity condition $P\equiv\Tr\r^2\leq 1$ in \cref{eq:rho-nice}, we also have 
\bes
\begin{align}
1 &\geq  P = \Tr\left[\left(\frac{1}{d}I + \vect{F}\cdot\vect{v}\right)^2\right] \\
&= \frac{1}{d} + \sum_{i,j=1}^M \Tr(F_iF_j)v_iv_j = \frac{1}{d}+\|\vect{v}\|^2\ ,
\end{align}
\ees
i.e., $\|\vect{v}\| = \sqrt{P-1/d}$, and \cref{eq:v-bounded} follows.

\section{Proof of \cref{eq:M-c,eq:v'}}
\label{app:cM-proof}

Using $\r' = \mcN(\rho) = \sum_\a K_\alpha\rho K_\alpha^\dag$ and the expansion $\r = \frac{1}{{d}} I + \sum_i v_i F_i$, we have the following series of implications:
\bes
    \begin{align}
        \frac{1}{{d}} I + \sum_i v_i'F_i &=  \sum_\a K_\alpha(\frac{1}{{d}}I+\sum_i v_iF_i) K_\alpha^\dag\\
        \sum_i v_i'F_i &= \frac{1}{{d}}\sum_\a K_\alpha K_\alpha^\dag + \sum_{\a i} v_i K_\alpha F_i K_\alpha^\dag-\frac{1}{{d}}I\\
        \sum_i v_i'\Tr(F_jF_i) &=\frac{1}{{d}}\sum_\a \Tr(F_j K_\alpha K_\alpha^\dag)  + \sum_{\a i} v_i\Tr( F_jK_\alpha F_i K_\alpha^\dag)\\
v_j' &= c_j + \sum_i M_{ji}v_i.
    \end{align}
    \label{eq:OSR-i}
\ees
\cref{eq:M-c,eq:v'} now follows from \cref{eq:OSR-i}.

\section{Proof of \cref{lemma1}}
\label{app:lemma1-proof}

Recall that the $p$-norm is defined in \cref{eq:Schatten-p}. The matrix H\"older inequality states that for $1\le a,b \le \infty$ and $\frac{1}{a}+\frac{1}{b}=1$~\cite{Bhatia:book,Baumgartner:11} :
\beq
\<A,B\> \le \|A\|_a \|B\|_b.
\eeq
An important special case for our purposes is $a=b=2$, i.e.:
\beq
\<A,B\> \le \sqrt{\<A,A\>\<B,B\>},
\label{eq:CS}
\eeq
which is just the Cauchy-Schwarz inequality for matrices.

Any linear map $\Psi:\mc{B}(\mcH)\mapsto\mc{B}(\mcH)$ on operators $X\in\mc{B}(\mcH)$ can be written as $\Psi (X)= \sum_\a E_\a X E_\a^{\prime \dag}$, where $\{E_\a,E'_\a\}\in\mc{B}(\mcH)$. Its Hermitian conjugate 
\beq
\<\Psi^\dag(X),Y\> = \<X,\Psi(Y)\> , 
\label{eq:supop-hc}
\eeq
can be written explicitly as $\Psi^\dag(X) = \sum_\a E_\a^\dag X E_\a^{\prime}$~\cite{ODE2QME}.

Therefore, if $\mcN = \{K_\a\}$ is a unital CPTP map, then so is $\mcN^\dag$. The reason is that 
if $\Ph$ is unital then $\Ph(X) = \sum_\a K_\a X K_\a^\dag$ and $\sum_\a K_\a K^\dag_\a = \sum_\a K_\a^\dag K_\a = I$. Thus $\Ph^\dag(X) = \sum_\a K^\dag_\a X K_\a$ has Kraus operators $\{L_\a = K_\a^\dag\}$, and it immediately follows that $\sum_\a L^\dag_\a L_\a = \sum_\a L_\a L^\dag_\a = I$, i.e., also $\Ph^\dag$ is a unital CPTP map. We can now prove \cref{lemma1}.

\begin{proof}
Consider $\r\in\mc{B}_+(\mcH)$ and let $P = \Tr(\r^2) = \<\r,\r\>$ denote its purity. The purity $P'$ of $\r^{(1)} = \mcN(\r) \equiv \mcN^{(0)}(\r)$ can be written as:
\bes
\begin{align}
P^{(1)} &= \<\r^{(1)},\r^{(1)}\> = \<\mcN(\r),\mcN(\r)\> = \< \r,\mcN^\dag [\mcN(\r)]\> \\
&= \< \r,\mcN^{(1)}(\r)\> \ ,
\end{align}
\ees
where in the third equality we used \cref{eq:supop-hc}.
Here, $\mcN^{(1)}\equiv \mcN^{(0)\dag} \circ \mcN^{(0)}$ is the composition of two CPTP maps, so $\mcN^{(1)}$ is itself a CPTP map. Thus $\r^{(2)} = \mcN^{(1)}(\r)$ is a quantum state [i.e., $\r^{(2)}\in\mcB_+(\mcH)$]. 

Define $\forall n\ge 1$ a sequence of quantum maps $\mcN^{(n+1)}\equiv \mcN^{(n)\dag} \circ\mcN^{(n)}$, purities $P^{(n)} = \< \r^{(n)},\r^{(n)}\>$, and states $\r^{(n+1)} = \mcN^{(n)}(\r)$. Then, using the Cauchy-Schwarz inequality for $n\ge 1$:
\bes
\label{eq:P2-bound}
\begin{align}
\label{eq:P2-bound-a}
P^{(n)} &= \< \mcN^{(n-1)}(\r),\mcN^{(n-1)}(\r) \> = \<\r , \mcN^{(n)}(\r)\> \\ 
\label{eq:P2-bound-b}
&\leq \<\r,\r\>^{1/2} \< \r^{(n+1)},\r^{(n+1)}\>^{1/2}\\
&= P^{1/2} (P^{(n+1)})^{1/2}\ ,
\end{align}
\ees
Expanding this recursion, we obtain:
\beq
P^{(1)} \le P^{1/2} P^{1/4} \cdots P^{1/2^n} (P^{(n+1)})^{1/2^n}  \ .
\label{eq:4.8.11}
\eeq
The purity is lower bounded by that of the fully mixed state $I/d$, where $d=\dim{\mc{H}}$: $P(I/d) = \Tr[(I/d)^2] = 1/d$. Therefore $\forall n,d$ we have $1/d \le P^{(n+1)} \le 1$ and hence $\lim_{n\to\infty} (P^{(n+1)})^{1/2^n} = 1$. Thus, upon taking the limit $n\to\infty$ of \cref{eq:4.8.11} we obtain:
\beq
P'  = P^{(1)} \le P^{\sum_{n=1}^\infty 2^{-n}} = P\ .
\eeq
Equality in \cref{eq:P2-bound-b} holds for all $\rho$ iff $\mcN^{(n)}(\r) = \r$ $\forall n$, i.e., $\mcN^{(n)} = \mc{I}$, which means that in particular, after setting $n=1$, $\mcN^{(0)\dag} \circ\mcN^{(0)} = \mc{I}$, so that by definition $\mcN^{(0)} = \mcN$ must be a unitary superoperator: $\mcN(\r) = \mc{U}(\r) = U\r U^\dag$, where $U$ is unitary.
\end{proof}

\section{Contractivity in the sense of \cite{Perez-Garcia:2006aa}}
\label{app:norm}

In \cref{sec:non-unital}, we define the contractivity of a map in terms of the Hilbert-Schmidt norm. This is different from \cite{Perez-Garcia:2006aa}, whose contractivity definition and results we briefly summarize here.

A map $\mc{N}$ between metric spaces $\mc{A}$ and $\mc{B}$ is strictly contractive iff there exists $r < 1$ s.t. $\forall A, B \in \mc{A}$ we have $d_{\mc{B}}(\mc{N}(A), \mc{N}(B)) \leq r d_{\mc{A}}(A, B)$, where $d$ is a distance function. If the above definition is satisfied with $r=1$ then the map $\mc{N}$ is called ``non-expansive''.

Let $\mathcal{M}_n$ be the space of $n\times n$ matrices, $\|\mathcal{N}\|_{p-p}=\sup_{A\in \mathcal{M}_n}\|\mathcal{N}(A)\|_p/\|A\|_p$ the induced $p$-norm, and $\|A\|_p$ (as usual) the Schatten $p$-norm of $A$. 
A positive trace-preserving map $\mathcal{N}:\mathcal{M}_n\rightarrow \mathcal{M}_{n'}$ is contractive when $\|\mathcal{N}\|_{p-p}<1$. For a non-unital map with $n=n'$, $\mathcal{N}$ is always non-contractive, i.e., $\|\mathcal{N}\|_{p-p}>1$ \cite{Perez-Garcia:2006aa}. The crucial difference from our case (and the main reason that this result does not contradict ours) is that $\mathcal{M}_n$ is allowed to contain the $0$ matrix (in addition, $\mathcal{N}$ need not be completely positive), which is, of course, different from the space of valid quantum states. Indeed, the proof of the non-contractivity of non-unital maps is essentially to take $A = 0$ and $B = I$; since $\mathcal{N}$ is trace-preserving, $\Tr(\mathcal{N}(B)) = \Tr(B) = n$ and if $\mathcal{N}(B) \neq I$ its norm must be larger than $B$'s.

\section{Proofs of \cref{lemma3} and \cref{lemma3-qubit}}
\label{app:lemma3-proof}

\subsection{Proof of \cref{lemma3}}
\begin{proof}
To prove \cref{eq:lemma3-a} note that, by definition, non-unital maps do not preserve $I$, i.e., they do not preserve the maximally mixed state given by the coherence vector $\vect{v}=\vect{0}$. The transformation $\vect{0} \rightarrow M\vect{0}+\vect{c}=\vect{c}$ must be non-zero. Hence, $\vect{c}\neq \vect{0}$.

To prove \cref{eq:lemma3-b}, note that 
\beq
\label{eq:12}
\|\vect{v}'\| = \|M\vect{v}+\vect{c}\|\le \sqrt{1-1/d} \ \ \forall \vect{v} \text{ s.t. }\|\vect{v}\|\le \sqrt{1-1/d},
\eeq 
where we used \cref{eq:v-bounded}. This must hold in particular for the maximally mixed state, i.e., when $\vect{v}=\vect{0}$. Hence, $\|\vect{c}\|\leq 1/\sqrt{1-1/d}$.

Since $\mcN$ is HS-contractive, \cref{eq:lemma3-c} follows directly from \cref{lem:M<1}.

To prove \cref{eq:lemma3-d}, recall the operator norm of $M$ is its maximum singular value:
\beq
\sigma_{\max}(M) = \|M\|=\sup_{\vect{v}\neq\vect{0}}\frac{\|M\vect{v}\|}{\|\vect{v}\|}.
\label{eq:op-norm}
\eeq
We use the definition in \cref{eq:op-norm} and write $1> \|M\|\ge \|M\vect{v}\|/\|\vect{v}\|, \forall \vect{v}\neq\vect{0}$, which implies that $\|M\vect{v}\|< \|\vect{v}\|, \forall \vect{v}\neq\vect{0}$.
\end{proof}

\subsection{Proof of \cref{lemma3-qubit}}

To prove that a single-qubit non-unital map is always HS-contractive, we will show that it satisfies $\|M\|<1$ and use \cref{lem:M<1} to complete the proof. 

\begin{proof}
Let $\vect{v}_M$ be a coherence vector satisfying $\|M\|=\|M\vect{v}_M\|/\|\vect{v}_M\|$. Without loss of generality we can assume $\vect{v}_M$ corresponds to a pure state: if it does not, we can normalize it so that $\|\vect{v}_M\|= 1/\sqrt{2}$.
We will prove the claim by contradiction and assume that $\|M\|\ge 1$, which implies $\|M\vect{v}_M\|=\|M\|\|\vect{v}_M\|\ge 1/\sqrt{2}$. Without loss of generality, let $(M\vect{v}_M)\cdot \vect{c}\ge 0$ (replace $\vect{v}_M$ by $-\vect{v}_M$ otherwise). This implies that $\|M\vect{v}_M+\vect{c}\|>\|M\vect{v}_M\|\ge 1/\sqrt{2}$. However, $M\vect{v}_M+\vect{c}$ is a valid quantum state which has $\|M\vect{v}_M+\vect{c}\|\le 1/\sqrt{2}$. Hence, by contradiction, $\|M\|<1$.

Using \cref{lem:M<1}, $\|M\|<1$ implies that the channel is HS-contractive. Hence, any single-qubit non-unital map is always HS-contractive.
\end{proof}

Note that while for $d=2$ positivity is captured entirely by the condition 
$\|\vect{v}\|\le 1/\sqrt{2}$, a condition on $\|\vect{v}\|$ alone is insufficient to ensure positivity for $d>2$; additional constraints must be satisfied (see, e.g., Ref.~\cite[Eqs.~(23) \& (24)]{PhysRevA.68.062322}). Since $\|\vect{v}_M\| \leq \sqrt{1 - 1/d}$ is not the only condition for $\vect{v}_M$ to represent a valid state, a coherence vector $\mathbf{v}_M$ such that $\|\vect{v}_M\| = \sqrt{1 - 1/d}$ could yield an invalid state that is non-positive. This invalidates the proof for general $d$. Hence, the proof above holds only for $d=2$.

\begin{figure*}
        \includegraphics[width=0.33\textwidth]{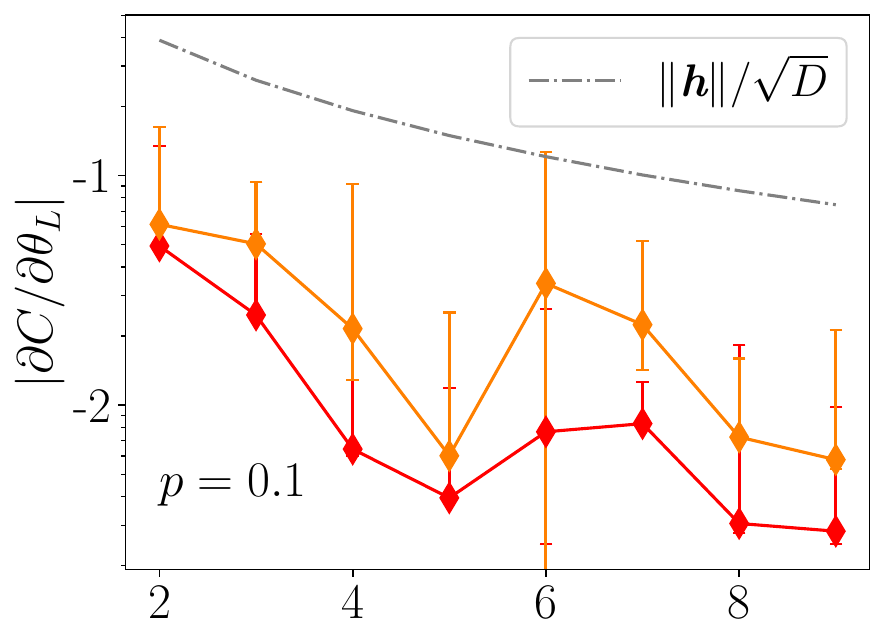}  
        \includegraphics[width=0.32\textwidth]{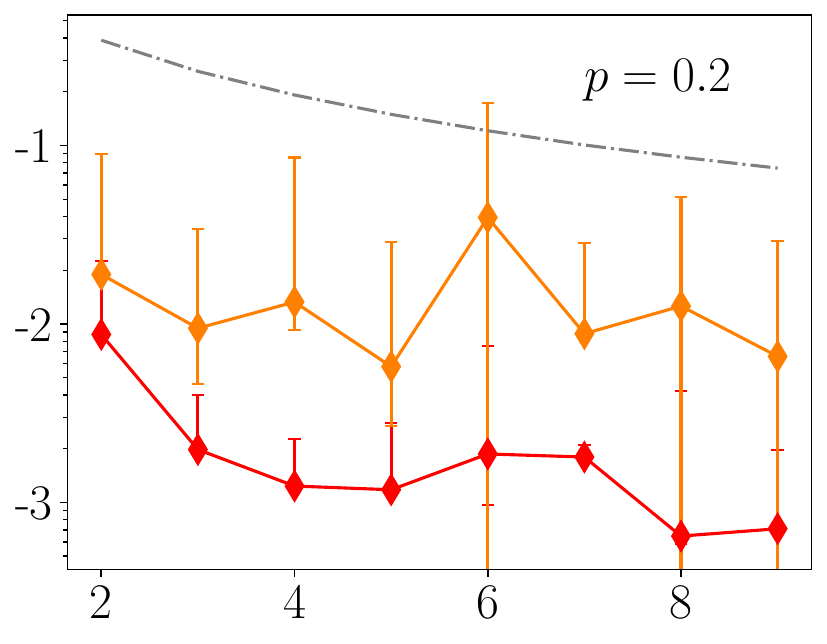}
        \includegraphics[width=0.32\textwidth]{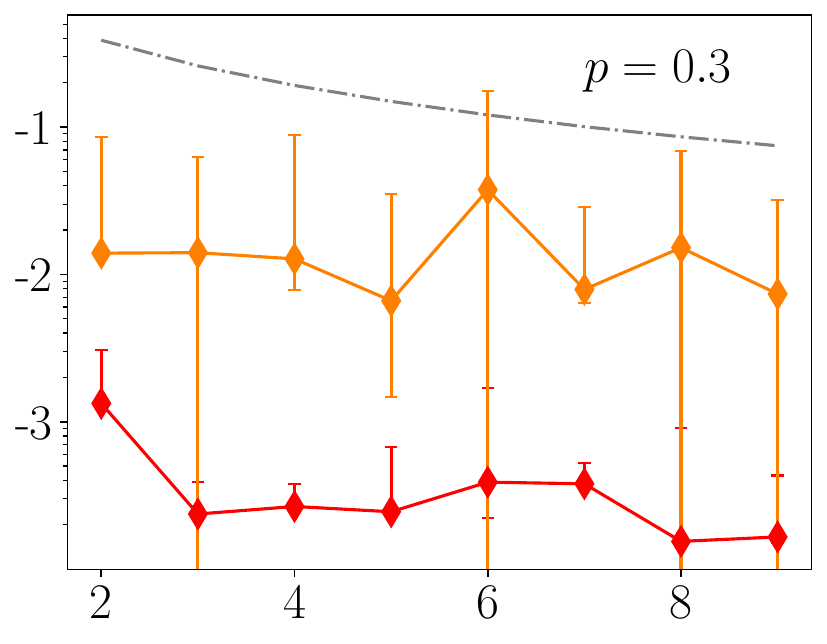}
        \includegraphics[width=0.33\textwidth]{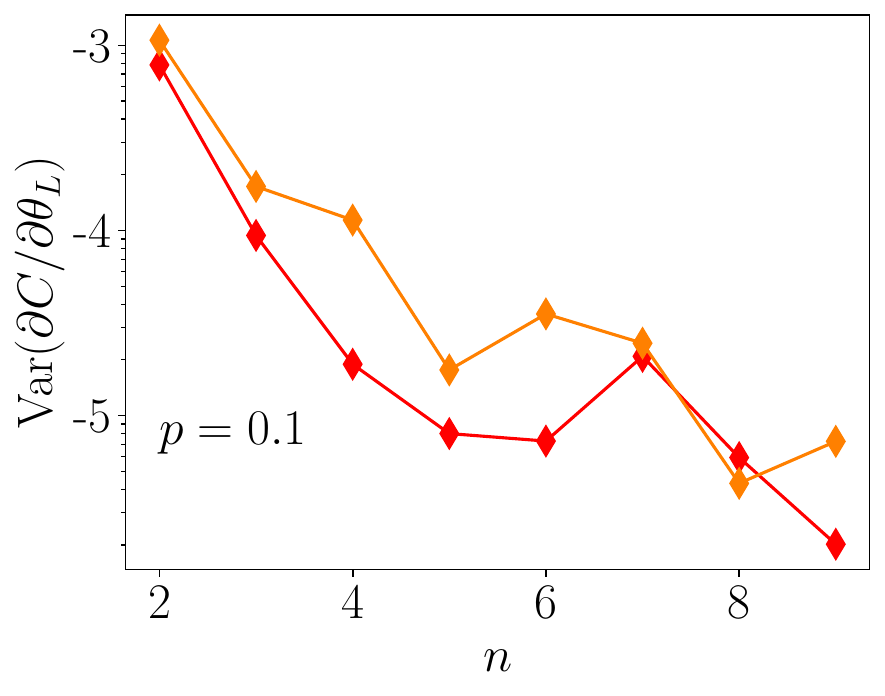} 
        \includegraphics[width=0.32\textwidth]{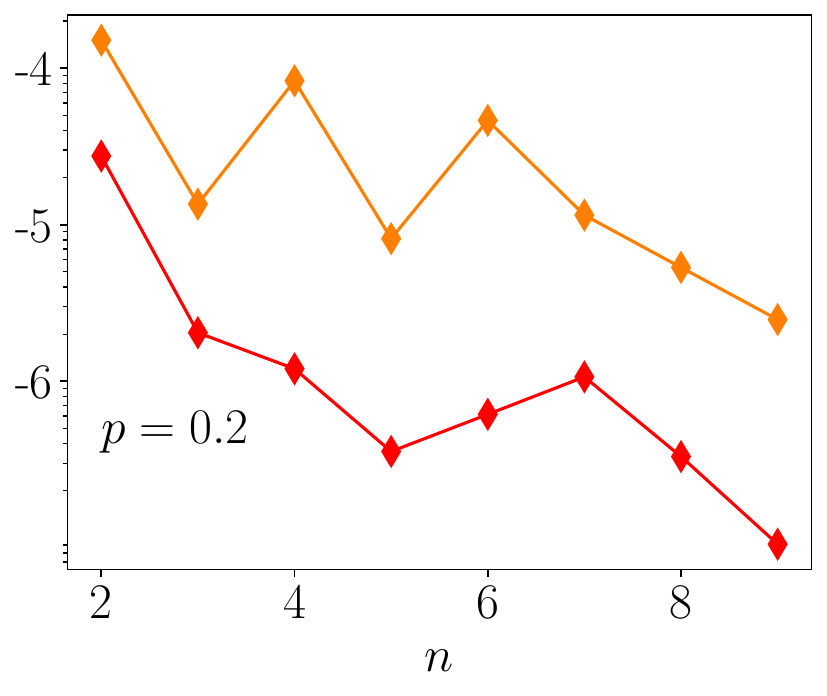}  
        \includegraphics[width=0.32\textwidth]{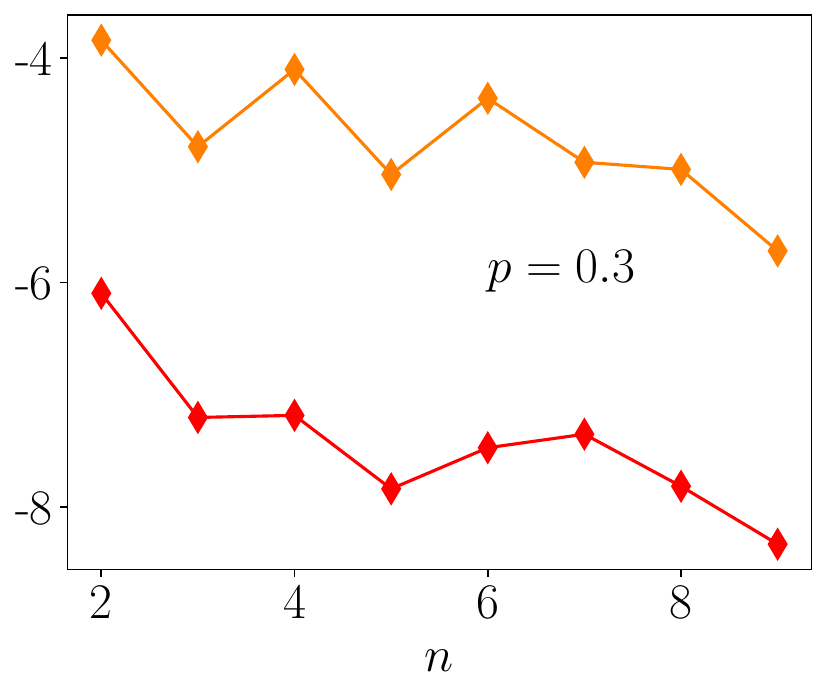}  
    \caption{Magnitude and variance of the gradient of the cost functions for depolarizing (red) and amplitude-damping (orange) maps as a function of the number of qubits for noise probabilities $p=0.1$, $0.2$, and $0.3$ in the left, middle, and right plots, respectively. Error bars in the top row represent the range of the values.} 
    \label{fig:nbound}
\end{figure*}

\section{Proof of \cref{eq:random-N-bound}}
\label{app:proof_of_random-N-bound}

Using \cref{eq:21c}, the noisy cost function is:
\bes
\begin{align}
    &C'(\vect{\theta}) = \Tr[\tilde{H}  \mcV(\theta_\mu)(\tilde{\rho})] \\
    &\quad = p_j\Tr[\tilde{H}  \mcU(\theta_\mu)(\tilde{\rho})]+ \sum_{k\neq j}p_k\Tr[\tilde{H} \mcV'_k(\theta_\mu)(\tilde{\rho})].
\end{align}
\ees
Using \cref{eq:cost-diff}, we now have:
\bes
\begin{align}
    &\frac{\partial C'(\vect{\theta})}{\partial \theta_\mu}
    = p_j\frac{\partial C(\vect{\theta})}{\partial \theta_\mu} 
    -\frac{i}{2}\sum_{k\neq j}p_k \Tr[\tilde{H} \mcV'(\theta_\mu)([P_k,\tilde{\rho}]) ]\\
    &= p_j\frac{\partial C(\vect{\theta})}{\partial \theta_\mu}+\frac{1}{2}\sum_{k\neq j}p_k \Tr[\tilde{H} \mcV'(\theta_\mu)\times\\
    &\qquad\qquad\qquad\qquad\left(\mcU_k\left(\frac{\pi}{2}\right)(\tilde{\rho}) -\mcU_k\left(-\frac{\pi}{2}\right)(\tilde{\rho}) \right) ]\notag \\
    &= p_j\frac{\partial C(\vect{\theta})}{\partial \theta_\mu} +\frac{1}{2}\sum_{k\neq j}p_k \Tr(\tilde{H} \tilde{\xi}_k),
\end{align}
\ees
where $\tilde{\xi}_k=\tilde{\rho}^+_k-\tilde{\rho}^-_k$ and $\tilde{\rho}_k^{\pm}=\mcV'(\theta_\mu)\mcU_k\left(\pm\frac{\pi}{2}\right)(\tilde{\rho})$. Thus,
\beq
    \left|\frac{\partial C'(\vect{\theta})}{\partial \theta_\mu}\right|
    \le p_{j(\mu)}\left|\frac{\partial C(\vect{\theta})}{\partial \theta_\mu}\right| +\frac{1}{2}\sum_{k\neq j}p_{k(\mu)} |\tilde{\vect{w}}_{k(\mu)}\cdot \vect{h}| .
    \eeq
This directly yields the bound on the gradient given in \cref{eq:random-N-bound}.

\section{Proof of \cref{eq:d.h-bound}}
\label{app:NILS}

\begin{proof}
Let us reproduce \cref{eq:d_L-again} for convenience:
    \begin{align*}
        \vect{d}_L = p_{1}^{L-1}\Theta_1\vect{c}_1+p_{2}^{L-2}\Theta_2\vect{c}_2+\cdots+p_{L-1}\Theta_{L-1}\vect{c}_{L-1}+\vect{c}_{L} .
    \end{align*}
Rotating a vector does not change its norm, i.e., $\|\Theta_i\vect{c}_i\|=\|\vect{c}_i\|$. Recall that $\|\vect{c}_i\|\le 1/\sqrt{1-1/d}$ from \cref{lemma3}. Thus, using $p \equiv \max_i p_i$ and $0\le p_i < 1, \forall i$:
\bes
    \begin{align}
        \|\vect{d}_L\|&=\| p_{1}^{L-1}\Theta_1\vect{c}_1+\cdots+p_{L-1}\Theta_{L-1}\vect{c}_{L-1}+\vect{c}_{L}\|\\
        &\le \| p_{1}^{L-1}\Theta_1\vect{c}_1\|+\cdots+\|p_{L-1}\Theta_{L-1}\vect{c}_{L-1}\|+\|\vect{c}_{L}\|\\
        &\le \| p^{L-1}\Theta_1\vect{c}_1\|+\cdots+\|p\Theta_{L-1}\vect{c}_{L-1}\|+\|\vect{c}_{L}\|\\
        &\le  (p^{L-1}+\cdots+p+1)\max_l \|\vect{c}_{l}\|\\
        &\le  (p^{L-1}+\cdots+p+1)\frac{1}{\sqrt{1-1/d}}\\
        &=\frac{1-p^L}{1-p}\frac{1}{\sqrt{1-1/d}} .
    \end{align}
\ees

Using the triangle inequality, we have
\begin{equation}
    \begin{aligned}
        |\vect{d}_L\cdot\vect{h}|&\le  \|\vect{d}_L\|\|\vect{h}\|\\
        &\le  \frac{1-p^L}{1-p}\frac{\|\vect{h}\|}{\sqrt{1-1/d}} \le \frac{1}{1-p}\frac{\|\vect{h}\|}{\sqrt{1-1/d}} ,
    \end{aligned}
\end{equation}
which is \cref{eq:d.h-bound}.
\end{proof}

\section{Dependence on circuit width}
\label{app:BD-proof}

In \cref{eq:non-unital-cost}, we considered the overlap of two randomly chosen $D$-dimensional vectors, which suggests that the gradient of the cost function scales as $1/\sqrt{D}$. This phenomenon was initially discussed in Ref.~\cite{McClean2018} and is the original (noise-free) barren plateau (BP). To rederive it, we follow the approach of Ref.~\cite{proof-random-vec-ortho}. 

Consider two normalized $D$-dimensional vectors $\vect{v}$ and $\vect{h}$, and choose each component of $\vect{h}$ uniformly from the surface of a normalized $D$-Ball, i.e., $\|\vect{v}\|=\|\vect{h}\|=1$. Without loss of generality, we can construct $\vect{h}$ such that its elements $h_i$ are chosen randomly from $\{-1,1\}/\sqrt{D}$, for $i\in[1,D]$.
The expectation value of the inner product of $\vect{v}$ and $\vect{h}$ is $\mathbb{E}[\vect{v}\cdot\vect{h}]=\mathbb{E}[\sum_i v_i h_i]=0$. The variance of their inner product is $\sigma^2[\vect{v}\cdot\vect{h}]=\mathbb{E}[\sum_{i,j} v_i h_i v_j h_j]-\mathbb{E}[\vect{v}\cdot\vect{h}]^2=\sum_{i,j}v_iv_j\mathbb{E}[h_ih_j]=\sum_{i} v_i^2/D=1/D$. The second to last equality uses $h_i^2=1/D$ and $\mathbb{E}[h_ih_j]=0$ for $i\neq j$.

The Chernoff bound states that $\Pr(|X|>\epsilon)<\exp(-\frac{\epsilon^2}{\sigma[X]^2})$. Applying this bound to \cref{eq:non-unital-cost}, we have $\Pr\left(2\left| \frac{\partial C(\vect{\theta})}{\partial \theta_{Lm}}\right|>\epsilon\right) = \Pr(|\vect{v}\cdot\vect{h}|>\epsilon) <e^{-D\epsilon^2}$. Substituting $\epsilon=1/\sqrt{D}$, we obtain $\Pr(|\vect{v}\cdot\vect{h}|>1/\sqrt{D})<1/e$, ensuring that it is likely that $|\vect{v}\cdot\vect{h}|$ stays below $1/\sqrt{D}$.

Note that since our $n$-qubit VQA Hamiltonian is $2$-local, the effective dimension $D$ of $\vect{h}$ is $\sum_{k=1}^2 {n\choose k}=(n^2+n)/2$ as discussed below \cref{eq:non-unital-cost}. Thus, this result can be interpreted as stating that the cost function gradient scale as the inverse of the number of qubits (circuit width).

Next, we examine whether this dependence on circuit width can be observed in numerical simulations of the same type as discussed in \cref{sec:sim}. We again employ a set of $50$ randomly chosen $n$-qubit Hamiltonians, with $2\le n\le 9$.  As stated in \cref{algo:randH}, we set $\|H\|_2=1$; this ensures that $\|\vect{h}\|\le 1$, i.e., does not grow with the effective dimension $D$ of $\vect{h}$. 
\cref{fig:nbound} shows the result of the simulation. 
The two types of noise we simulated exhibit similar patterns, with amplitude-damping having a larger gradient magnitude and variance. Our simulation results exhibit a discernible trend along the dashed-dotted gray line denoted as $\|\vect{h}\|/\sqrt{D}$. The magnitude of the gradient (top row) appears to approach the scaling $1/\sqrt{(n^2+n)/2}$ as the noise probability increases. 
We emphasize that this does not constitute an upper bound for the magnitude; rather, it represents an expected value obtained through averaging over a large number of randomly generated vectors.

\section{Trainability of HS-contractive non-unital circuits}
\label{app:trainable}

We will refer to a circuit as trainable when the variance of its cost function gradient vanishes no faster than $\Omega(1/\text{poly}(n))$~\cite{Cerezo2021-bp}.

Using \cref{eq:final-bound-2} and letting $q=\min_{i\in [l,L]}\sigma_{\min}(M_i)>0$, we write
\beq
\|\tilde{\vect{v}}_\mu^L\|\ge q^{L-l}d_l,
\eeq
where $\mu=(l,m)$, $d_l$ is a constant [recall the argument between \cref{eq:norm-d_l-1-bound,eq:norm-d_l-1-bound-2}], and $0<q<1$ due to $M_i$ being HS-contractive. We dropped the term $2p^L$ in \cref{eq:final-bound-2} since it becomes negligible in the large $L$ limit. We can rewrite \cref{eq:gradient-result} as
\bes
\label{eq:H2}
\begin{align}
    |\partial_\mu C| &=\frac{1}{2}|\tilde{\vect{v}}_\mu^L\cdot\vect{h}|\\
    &=\frac{1}{2}\|\tilde{\vect{v}}_\mu^L\|\|\vect{h}\||\cos(\theta)|\\
    &\ge q^{L-l}\epsilon .
\end{align}
\ees
Here $\epsilon=d_l\|\vect{h}\||\cos(\theta)|$ and 
$\cos(\theta)=(\tilde{\vect{v}}_\mu^L\cdot\vect{h})/\|\tilde{\vect{v}}_\mu^L\|\|\vect{h}\|$ is the BP factor unrelated to noise, i.e., the overlap between two \emph{normalized}, $D$-dimensional, random vectors that stays below $1/\sqrt{D}$ as argued in  \cref{app:BD-proof}.  
The difference is that the norm of $\tilde{\vect{v}}_\mu^L$ now scales down as $q^{L-l}$, so there is an additional noise-induced effect. 

To account for the variance of the cost function gradient, we reinterpret \cref{eq:H2} as a statement about an ensemble of random quantum circuits following the ansatz of the form in \cref{eq:1}, so that $|\partial_\mu C|$ becomes a random variable. Then, recalling Chebyshev's inequality
\beq
\Pr(|\partial_\mu C|\ge \delta)\le \frac{\text{Var}[\partial_\mu C]}{\delta^2},
\eeq
we can use \cref{eq:H2} to write
\beq
\text{Var}[\partial_\mu C]\ge q^{2(L-l)}\epsilon^2 .
\eeq
Therefore, to satisfy the trainability condition $\text{Var}[\partial_\mu C]\ge \Omega(1/\text{poly}(n))$ we require
\beq
L-l=O(\log(n)) .
\eeq
\cref{th:non-unital-no-NIBP} states that a circuit for which the last $O(1)$ layers are HS-contractive non-unital does not suffer from an NIBP. This is true also for the last $O(\log(n))$ layers since \cref{eq:final-bound-2} is still $1/\text{poly}(n)$-small in this case. This implies that the cost function within the last $O(\log(n))$ HS-contractive non-unital layers is trainable, in qualitative agreement with Ref.~\cite{mele2024noiseinduced}, who used a different layer-independent constant.

\section{Example of non-unital channel outside \cref{th:non-unital-no-NIBP}} \label{app:sigma_min}
We expand on the example of a non-unital channel with $\sigma_{\max}(M)>0$ and $\sigma_{\min}(M)=0$, which was mentioned below the proof of \cref{th:non-unital-no-NIBP}. This kind of channel is expected to avoid NIBPs using a Levy's lemma-type argument.

Consider a composite of a bit flip (BF) channel followed by an amplitude damping (AD) channel. The Kraus operators of an AD channel are given in \cref{eq:AD}, while those of a BF channel are
\begin{equation}
    K_0=\sqrt{p}I,\quad K_1=\sqrt{1-p}X,
    \label{eq:BF}
\end{equation}
where $1-p$ is the probability that a bit flip occurs. We can also find, using \cref{eq:M-c}, that $\vect{c} = \vect{0}$ (as is true for all unital channels) and $M = \text{diag}(1,2p-1,2p-1)$. When $p=1/2$, we have $\sigma_{\max}(M) = 1$ and $\sigma_{\min}(M) = 0$. 

Recall that a composition $\Phi(\rho)=\Psi_2(\Psi_1(\rho))$ of two CP maps $\Psi_1(\rho)=\sum_\alpha J_\alpha \rho J_\alpha^\dag $ and $\Psi_2(\rho)=\sum_\beta K_\beta \rho K_\beta^\dag$ is:
\begin{equation}
        \Phi(\rho)=\sum_\gamma L_\gamma \rho L_\gamma^\dag,
\end{equation}
where $\gamma = (\alpha,\beta)$ and $L_\gamma = K_\beta J_\alpha$ are the Kraus operators of the composite channel. 

Using $\Psi_1$ as a BF channel with $p=1/2$ and $\Psi_2$ as an AD channel in \cref{eq:AD}, we can construct $L_\gamma$ as
\begin{equation}
    \begin{aligned}
        &L_{(0,0)}=\begin{pmatrix}
            1/\sqrt{2} & 0\\
            0 & \sqrt{(1-p)/2}
        \end{pmatrix} &&
        L_{(1,0)}=\begin{pmatrix}
            0 & \sqrt{p/2} \\
            0 & 0
        \end{pmatrix}\\
        & L_{(0,1)}=\begin{pmatrix}
            0 & 1/\sqrt{2} \\
            \sqrt{(1-p)/2} & 0
        \end{pmatrix}&&
        L_{(1,1)}=\begin{pmatrix}
            \sqrt{p/2} & 0 \\
            0 & 0
        \end{pmatrix}.
    \end{aligned}
\end{equation}
This set of Kraus operators corresponds to $M=\text{diag}(\sqrt{1-p},0,0)$ and $\vect{c}=(0,0,p)$. We have composed a non-unital channel with $\sigma_{\max}(M)=\sqrt{1-p}$ and $\sigma_{\min}(M)=0$.

Alternatively, we can derive this using the transformation of the coherence vector $\vect{v}$, by a BF channel with $p=1/2$ followed by an AD channel:
\begin{equation}
    \begin{aligned}
        \vect{v}&\rightarrow M_{\text{AD}} M_{\text{BF}}(\vect{v}) + \vect{c}_{\text{AD}}\\
        &=M\vect{v}+\vect{c},
    \end{aligned}
\end{equation}
where $M=M_{\text{AD}} M_{\text{BF}}$ and $\vect{c}=\vect{c}_{\text{AD}}$ as expected.

Using $Z$ in place of $X$ in \cref{eq:BF} to first construct a dephasing channel with $p=1/2$ would also yield $\sigma_{\min}(M)=0$, as would other unital channels at special values of their parameters.

\end{document}